\documentclass[12pt]{article}
\usepackage{amssymb}
\usepackage{graphicx}
\usepackage{enumerate}
\usepackage{natbib}
\usepackage{url} 
\usepackage{xcolor}
\usepackage{hyperref}
\usepackage{movie15}
\usepackage{bbm}
\usepackage{tikz}
\usepackage{gensymb}
\usepackage{amsfonts}
\usepackage{amsmath}
\usepackage{units}
\usepackage{amsthm}
\usepackage{multicol}

\newtheorem{theorem}{Theorem}

\newtheorem{lemma}{Lemma}

\newtheorem{remark}{Remark}
\newcommand{\blind}{0}
\newcommand{\bb}{\boldsymbol\beta}

\def\arxiv{1}

\usepackage{xr-hyper}
\makeatletter
\newcommand*{\addFileDependency}[1]{
  \typeout{(#1)}
  \@addtofilelist{#1}
  \IfFileExists{#1}{}{\typeout{No file #1.}}
}
\makeatother
\newcommand*{\myexternaldocument}[1]{%
    \externaldocument{#1}%
    \addFileDependency{#1.tex}%
    \addFileDependency{#1.aux}%
}
\myexternaldocument{supp}

\usepackage{algorithm}
\usepackage{algpseudocode}

\usepackage{graphicx}
\usepackage{amsmath}
\usepackage{bm}
\usepackage{ragged2e}
\usepackage{makecell}
\usepackage{booktabs} 
\RequirePackage{amsthm,amsmath,amsfonts,amssymb}

\newcommand\x{\mathbf{x}}
\renewcommand\theadalign{cl}
\renewcommand{\cellalign}{cl}
\renewcommand\theadfont{\bfseries}
\renewcommand\theadgape{\Gape[2pt]}
\renewcommand\cellgape{\Gape[1pt]}
\newcommand\Po{\mathbf{P}_1}

\newcommand{\y}{\mathbf{y}}
\newcommand{\X}{\mathbf{X}}

\newcommand\nvar{P}

\begin{document}

\def\spacingset#1{\renewcommand{\baselinestretch}%
{#1}\small\normalsize} \spacingset{1}


%

%
%

\if0\blind
{
  \title{\bf TECHNICAL REPORT: \\ Sparse Bayesian Lasso via a Variable-Coefficient $\ell_1$ Penalty}
  \author{Nathan Wycoff$^1$, Ali Arab$^2$, Katharine M. Donato$^{3,4}$, Lisa O. Singh$^{1,5}$\hspace{.2cm}\vspace{0.5em}\\ 
    $^1$ The McCourt School's Massive Data Institute, Georgetown University\\
    $^2$ Department of Mathematics and Statistics, Georgetown University\\
    $^3$ The Institute for the Study of International Migration, Georgetown University\\
    $^4$ The School of Foreign Service, Georgetown University\\
    $^5$ Department of Computer Science, Georgetown University\\
    }
  \maketitle
} \fi

\if1\blind
{
  \bigskip
  \bigskip
  \bigskip
  \begin{center}
    {\LARGE\bf Sparse Bayesian Lasso via a Variable-Coefficient $\ell_1$ Penalty}
\end{center}
  \medskip
} \fi

\bigskip
\begin{abstract}
Modern statistical learning algorithms are capable of amazing flexibility, but struggle with interpretability.
One possible solution is \textit{sparsity}: making inference such that many of the parameters are estimated as being identically $0$, which may be imposed through the use of nonsmooth penalties such as the $\ell_1$ penalty.
However, the $\ell_1$ penalty introduces significant bias when high sparsity is desired.
In this article, we endow the $\ell_1$ penalty with learnable penalty weights $\lambda_p$.
We start by investigating the optimization problem this poses, developing a proximal operator associated with the $\ell_1$ norm.
We then study the theoretical properties of this variable-coefficient $\ell_1$ penalty in the context of penalized likelihood.
Next, we investigate application of this penalty to Variational Bayes, developing a model we call the Sparse Bayesian Lasso which allows for behavior qualitatively like Lasso regression to be applied to arbitrary variational models.
In simulation studies, this gives us the Uncertainty Quantification and low bias properties of simulation-based approaches with an order of magnitude less computation.
Finally, we apply our methodology to a Bayesian lagged spatiotemporal regression model of internal displacement that occurred during the Iraqi Civil War of 2013-2017.

\end{abstract}

\noindent%
\vfill

\newpage
\spacingset{1.9} 


\section{Introduction}\label{sec:intro}

\textbf{Basic Inferential Problem:} 
Sparsifying penalties are popular because they allow for simultaneous variable selection and parameter estimation by shrinking parameter values identically to $0$.
Explicitly, we define the sparse learning problem as follows:
\begin{equation}
\begin{split}
    \underset{\boldsymbol\beta,\boldsymbol\theta}{\min}\,\, \mathcal{L}(\boldsymbol\beta,\boldsymbol\theta) + \tau \Vert\boldsymbol\beta\Vert_0 \,\, ,
\end{split}
\end{equation}
where $\boldsymbol\beta$ is a parameter vector which we wish to be sparse, $\boldsymbol\theta$ is a vector of other model parameters and $\mathcal{L}$ is a misfit term which determines the compatibility of a parameter configuration with a dataset (such as a negative log-likelihood or log posterior density).
This is a difficult problem because the $||.||_0$ function, which counts the nonzero components of its input, is nonconvex and discontinuous.
The $||.||_1$ norm is a useful convex approximation to $||.||_0$.
This is the penalty function of the celebrated Lasso \citep{Tibshirani1996,taylor1979deconvolution}, which may be viewed as placing a Laplace prior on $\beta_p$ \citep{Park2008}.
The draw-back to this convexity is bias: estimates of nonzero parameters will be shrunk towards zero, significantly so if we are to achieve much sparsity \citep{Zhang2010}. 
One solution to this problem is to allow the penalty coefficients to vary by parameter:
\begin{equation}
    \underset{\boldsymbol\beta,\boldsymbol\theta}{\min}\,\, \mathcal{L}(\boldsymbol\beta, \boldsymbol\theta) + \tau \sum_{p=1}^P \lambda_p |\beta_p| \,\, ,
\end{equation}
as when $\lambda_p\to0$ the associated $\beta_p$ is unpenalized. 
Far from clear is how to choose these $\lambda_p$, as evidenced by the assortment of approaches discussed in the next section.

\noindent\textbf{Existing Approaches and Limitations:}
The importance of variable selection in nonlinear models has prompted decades of research and yielded a multitude of approaches.  
In this section, we review selected penalty-based approaches to variable selection.

Some approaches perform adaptation ``offline", that is, before or after the optimization process.
The Adaptive Lasso \citep{Zou2006} specifies $\lambda_p=\frac{1}{|\hat{\beta}_p|^\gamma}$, where $\hat{\beta_p}$ is some initial estimate of the regression coefficient and $\gamma$ is a hyperparameter, though an initial estimate $\hat{\beta}_p$ is a nontrivial ask if $\mathcal{L}$ is complicated.
\cite{buhlmann2008,candes2008enhancing} propose to iterate this procedure, updating the penalty coefficient $\lambda_p$ with new $\hat{\beta_p}$.
Alternatively, we may use Lasso for variable selection only then proceed to an unpenalized procedure \citep{Efron2004,meinshausen2007relaxed,Zou2008}.

In the Bayesian framework, several authors have proposed simulation-based procedures for adaptation of $\lambda_p$.
\cite{kang2009self,leng2014bayesian,mallick2014new} place a Gamma prior or similar for $\lambda_p$ and conduct inference via Gibbs sampling.
\cite{bhattacharya2012bayesian} instead take the approach of specifying a Dirichlet prior on the regression coefficients before conducting MCMC.
Conceptually this is similar to the Horseshoe Prior \citep{Carvalho2010}, which instead specifies a conditional Normal prior for $\beta_p $.
The Horseshoe prior enjoys generality as well as fast specific implementations \citep{Terenin2019,Makalic2015}.
The reader is referred to \cite{Bhadra2019} for an extensive comparison of the Lasso and Horseshoe models.
Another approach is the Spike-Slab prior \citep{Mitchell1998} which uses discrete latent variables to categorize whether $\beta_p$ was \textit{a priori} sampled from the spike or the slab, which would complicate gradient-based inference. 

Variational Bayesian inference \citep[VB]{blei2017variational} is another way to adapt the $\lambda_p$.
But there is a sense in which VB with a Laplace prior differs from Lasso: the integrated cost erases the sparsifying geometry, as we discuss in Section \ref{sec:vb}.
Previous work has circumvented this using either a point mass variational distribution for regression coefficients \citep{tung2019bayesian}, which does not directly allow for uncertainty quantification (without, say, bootstrapping \citet{fu2000asymptotics}), or by thresholding small coefficients \citep{babacan2014bayesian}, which requires setting an arbitrary threshold \citep{she2009thresholding}.
\cite{kawano2015predictive} vary the threshold while monitoring information criteria.

Nonconvex penalties do not vary the regularization strength but are directly constructed to impose minimal bias on nonzero coefficients, such as the Smoothly Clipped Absolute Deviation function \citep[SCAD]{Fan2001,Hunter2005} or Minimax Concave Penalty \citep[MCP]{Zhang2010}.
These penalties do not correspond to proper priors, as they place constant, positive probability density arbitrarily far from the origin.

Other penalties, like the ridge penalty \citep{hoerl1970ridge} and elastic net \citep{zou2005regularization}, do not serve primarily to select variables, but rather to stabilize model fit or improve predictions. These are not the subject of this article.

Somewhat in between Bayesian and penalized likelihood approaches lies Sparse Bayesian Learning \citep{tipping2001sparse}.
Here, sparsity comes not from nonsmooothness of the likelihood with respect to $\beta$, but rather by shrinking prior variance terms (and hence posterior variance terms) to zero via empirical Bayes.
As originally proposed, this requires a linear model and Gaussian error structure, though its reach can be expanded using Gaussian mixtures \citep{sandhu2021nonlinear}.
\citet{helgoy2019sparse} applied this framework to the Bayesian Lasso.

\noindent\textbf{Scientific Motivation:} 
Forced human migration is at an all time high and only increasing.
In order for policymakers to plan effectively, it is important to understand which factors affect when and where people will move during times of crisis.
Traditionally, migration researchers use \textit{gravity models}, described in Section \ref{sec:gravity}. These models were developed by economist to model international trade (the flow of goods and services) as a function of traditional econometric indicators, typically on a coarse spatiotemporal scale (i.e. country by year). Migration researchers use this model because it can be effective in capturing both push and pull dynamics. 
In this article, we study an approximately year-long period of the Iraqi Civil War of 2013-2017, where the government and its international backers fought ISIL (Islamic State, also known as ISIS) in the northeast of the country.
Traditional migration variables are not immediately available in exactly those situations where they are needed most: dangerous, remote, and underdeveloped areas, such as our case study. 
Taking advantage of the proliferation of the mobile internet, we develop novel social media ``Buzz Variables" derived from Arabic-Language twitter (Section \ref{sec:data}).
This allows us to make predictions at finer spatiotemporal scales than traditional country-year level administrative data, but introduces new challenges, namely unignorable lag between cause and effect, spatiotemporal correlation, and high dimensional and noisy predictors.
In order to meet these new challenges, we develop a Bayesian hierarchical model which allows for zero inflation, overdispersion, spatiotemporal random effects, and lagging and aggregation of predictors (Section \ref{sec:app_model}) -- a complex model to perform fixed-effects selection for.

\noindent\textbf{Brief Outline of Our Contributions:}
We begin this article with a study of the approach of treating the Laplace inverse scale parameter $\boldsymbol\lambda:=\{\lambda_1,\ldots,\lambda_P\}$ as an additional parameter to be optimized, one endowed with the hyperprior $p_\lambda$, yielding:
\begin{equation}
    \underset{\beta,\theta,\boldsymbol\lambda>\mathbf{0}}{\min}\,\, \mathcal{L}(\beta,\theta) + \sum_{p=1}^P \big[\tau\lambda_p |\beta_p| - \log\lambda_p\big] + \sum_{p=1}^P -\log p_\lambda (\lambda_p) \,\,,
\end{equation}
where $\tau$ is a positive parameter controlling the sparsity level and the  $\log\lambda_p$ term comes from the normalization constant of the Laplace density.
The $\ell_1$ penalty in general may be efficiently applied to a smooth loss function via the Iterative Shrinkage and Thresholding Algorithm \citep[ISTA]{daubechies2004iterative}, a proximal gradient method \cite{parikh2014proximal}.
But ISTA assumes a known and fixed $\lambda_p$.
Section \ref{sec:prox} examines a proximal operator associated with the variable-$\lambda_p$ optimization problem and discusses how to deploy it to marry adaptive and nonsmooth penalties in an analog of ISTA we term Variable ISTA (VISTA) 
\footnote{As we deploy Nesterov acceleration, the better analog is actually Fast ISTA \citep{beck2009fast}.}.
It also investigates the basic theoretical properties of this procedure in penalized likelihood.

In Section \ref{sec:linear} we pivot to a Variational Bayesian take on an adaptive Lasso using a nonsmooth penalty we term the Sparse Bayesian Lasso (SBL).
The ``Bayesian Lasso" of \citet{Park2008} is so named as it is a Bayesian explanation for the density used in Lasso.
But the Laplace prior does not actually encode the prior beliefs that motivate the use of Lasso, i.e., that the coefficient vector is sparse.
Rather, it encodes the belief that most of the coefficients are near zero, which leads to poor interval coverage and bias if they are not.
As such it is not \textit{operationally} a Bayesian Lasso, that is, it is not a method that yields the sparsity properties of Lasso as part of a Bayesian analysis.
The proposed Sparse Bayesian Lasso fills this role by allowing for both geometric sparsity and uncertainty quantification.
Combined with warm starts\footnote{A ``warm start" is simply the practice of initializing an iterative algorithm at the final step of a previous, similar algorithm.}, we can efficiently calculate trajectories not just of parameter estimates, as in the traditional Lasso, but of entire variational distributions, showing us how uncertainty of model parameters change with sparsity.
This allows us to avoid specifying or estimating an overall model complexity parameter $\tau$, which is difficult in practice.

In this article, we explain how to deploy our method in gradient-based learning frameworks such as \texttt{tensorflow} and \texttt{pytorch}.
After verifying the frequentist properties of the SBL in simulation studies, we demonstrate the potential of the SBL in complex hierarchical models using our Iraq forced displacement case study (Section \ref{sec:application}).
We conclude by discussing research directions suggested by the novel proximal operator in Section \ref{sec:conclusion}.

\section{Coefficient-Specific \texorpdfstring{$\lambda_p$}{lambda p} with Optimization}\label{sec:prox}

We begin this section by developing a proximal gradient method suitable for general smooth optimization problems augmented by the variable-coefficient $\ell_1$ penalty.
We then describe basic theoretical properties of the estimator in the penalized likelihood case.
We simply give outlines of proofs for theoretical results; detailed derivations are available in the Supplementary Material.


\subsection{Background: Nonsmooth Penalties and Proximal Operators}\label{sec:prox_bg}

We are interested in minimizing a complicated but smooth loss $l$ augmented with a simple but nonsmooth regularizer $g$:
\begin{equation}\label{eq:auglik}
    c(\x) = l(\x) + g(\x)\, ,
\end{equation}
where $g$ is given by the $\ell_1$ norm in the case of Lasso regression. Proximal gradient descent and its relatives are the algorithms of choice in such a situation. 
Given a function $g$ with domain $\mathcal{X}$ and some norm parameterized by a positive definite matrix $\mathbf{C}$, their \textit{proximal operator} is then defined as the following mapping:
\begin{equation}
    \mathrm{prox}_{g}^\mathbf{C}(\mathbf{x}) = \underset{\mathbf{u} \in \mathcal{X}}{\mathrm{argmin}} \, g(\mathbf{u}) + \frac{1}{2}||\mathbf{x}-\mathbf{u}||_{\mathbf{C}}^2 \,\, .
\end{equation}
Intuitively, the proximal operator of a function $g$ evaluated at a vector $\mathbf{x}$ returns another vector $\mathbf{u}$ which is close to $\mathbf{x}$ (wrt $\mathbf{C}$) but does a better job minimizing $g$.

Proximal Gradient algorithms optimize the objective function in Equation \ref{eq:auglik} via iteration of a two step process. Given a current solution $\mathbf{x}^k$ and smooth cost function $\mathcal{L}$, the next iterate is parameterized by a step size $s$ and given by:
\begin{align}\label{eq:prox_descent}
    \hat{\mathbf{x}}^{k+1} = \mathbf{x}^k - s \mathbf{C}^{-1}\nabla \mathcal{L}(\mathbf{x}^k) \\
    \mathbf{x}^{k+1} = \mathrm{prox}^{s\mathbf{C}}_{ g}(\hat{\mathbf{x}}^{k+1})\, .
\end{align}
Note that $\mathbf{C}$ is scaled by the step size in the proximal operator and that we have not included a subgradient of $g$ in the gradient descent step.
The matrix $\mathbf{C}$ may in theory be any positive definite matrix and in practice it is defined by the preconditioning strategy of the gradient descent algorithm that the proximal operator will be deployed in.



The proximal operator is most useful when it can be computed efficiently.
In this article, we will assume a diagonal preconditioner $\mathbf{C}=\mathrm{diag}(c_1,\ldots,c_P)$, which is conducive to breaking the proximal problem into subproblems defined along each axis.
For example, when $g:\mathbb{R}\to\mathbb{R}$ is given by $g(x)=\lambda|x|$, the proximal operator is given by elementwise application of the Soft Thresholding Operator (STO):
\begin{equation}\label{eq:soft.thresh}
    [\mathrm{prox}^{\mathrm{s diag(c_1,\ldots,c_P)}}_{\lambda |.|}(\mathbf{x})]_p = (|x_p|- s c_p \lambda)^+\mathrm{sgn}(x_p)\, ,
\end{equation}
where $(a)^+$ gives $\max(0,a)$ and $\mathrm{sgn}(a)$ gives the sign of $a$.


%
%

\subsection{The Variable-Coefficient \texorpdfstring{$\ell_1$}{l1} Proximal Operator}\label{sec:adapt_prox}

\begin{figure}
	\centering
	\raisebox{-0.5\height}{\includegraphics[scale=0.60]{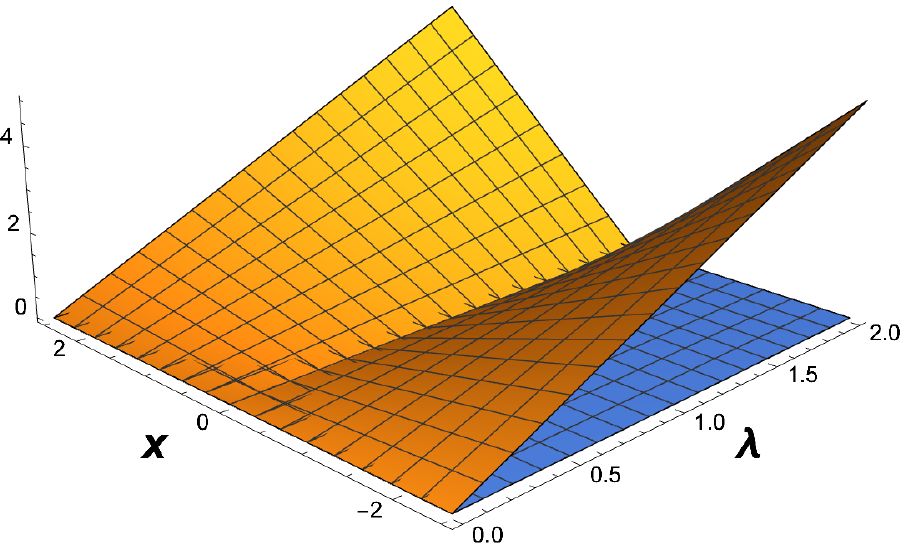}} \,\,
	\raisebox{-0.5\height}{\includegraphics[scale=0.52]{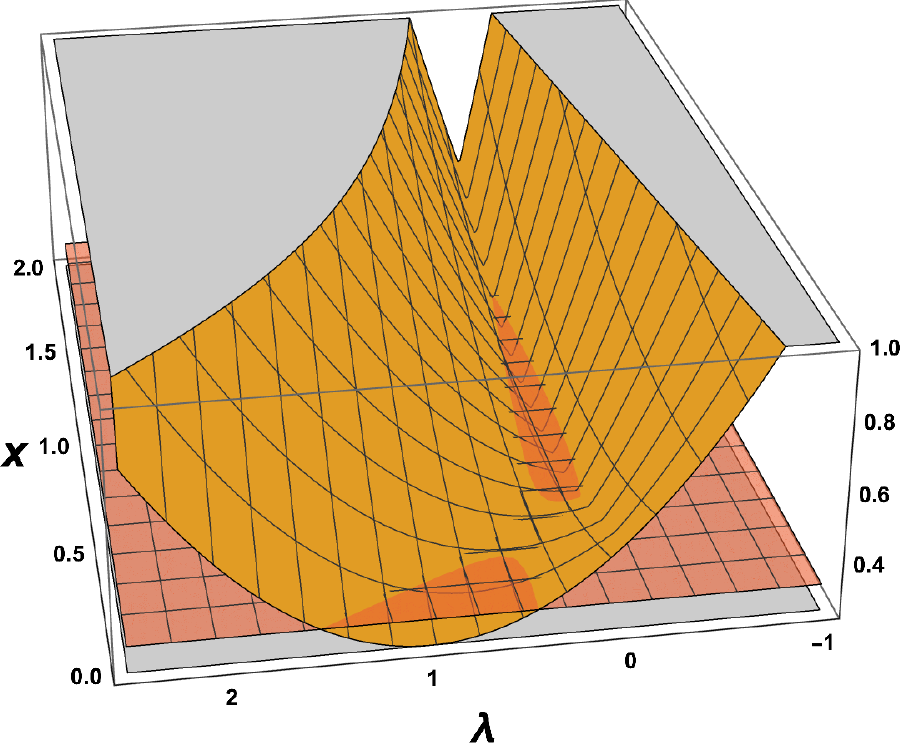}}
	\raisebox{-0.5\height}{\includegraphics[scale=0.52]{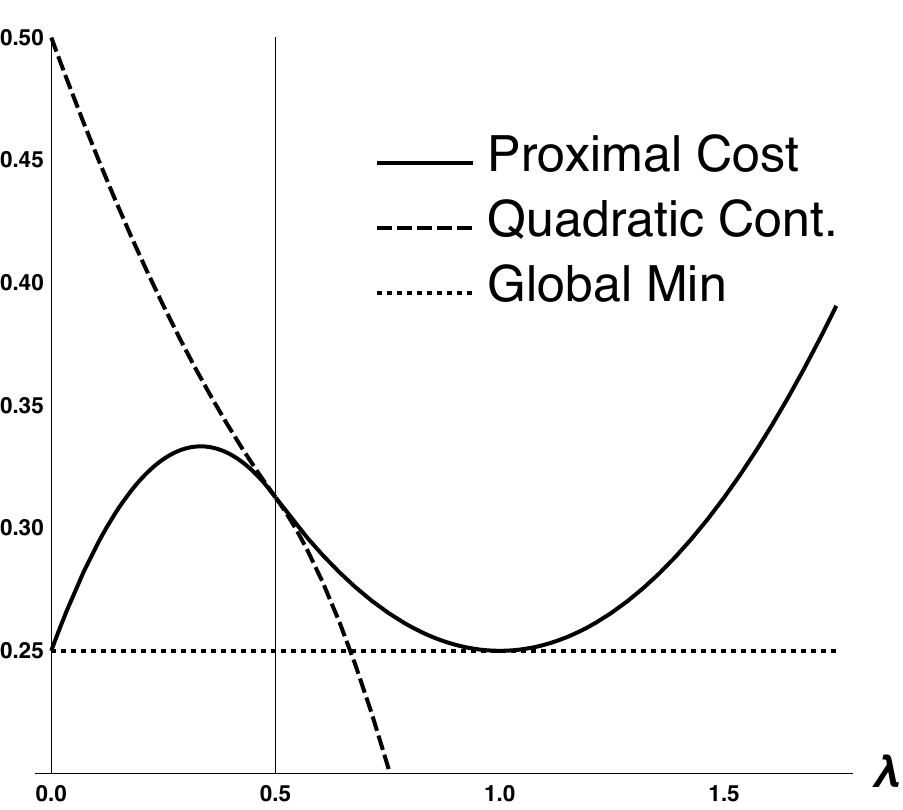}}
	\caption{\textbf{The Proximal Cost} 
	\textit{Left:} The function $g(x,\lambda)=\lambda|x|$.
	The proximal cost \textit{Center:} of $(x,\lambda)$
	\textit{Right:} marginal for $\lambda$ with $\lambda_0 = x_0=1; s_\lambda=s_x = 2$ yielding two optima.
	}
	\label{fig:nonconv}
\end{figure}

Because the $\ell_1$ regularization coefficient $\lambda$ is now being optimized over, the STO is no longer the pertinent proximal operator. 
Indeed, the $\ell_1$ regularization function, considered formally as a function of both $\lambda$ and $\mathbf{x}$, is nonconvex (see Figure \ref{fig:nonconv}, Left), somewhat complicating proximal operator computation.
In fact, many authors such as \citet{parikh2014proximal} define the $\mathrm{prox}$ operator as one that acts on convex functions.
However, there has been work on extending proximality to larger classes of functions (see e.g. \citep{hare2009computing}) and, of particular interest to the statistical community, development of proximal operators for the nonconvex ``Bridge Penalties" $|\beta_p|^q$ for $q\in(0,1)$ \citep{marjanovic2013exact}; see \cite{polson2015proximal} for more on proximal methods in statistics.
Perhaps because of this focus on convexity (and despite the popularity of the $\ell_1$ norm and adaptive penalty methods), the proximal operator of $\lambda |x|$ as a $\mathbb{R}^+\times\mathbb{R}\to\mathbb{R}^+$ function has not to our knowledge been previously examined in the literature.
It turns out that the action $\mathrm{prox}(\mathbf{x},\boldsymbol\lambda)$ of this proximal operator is available in closed form and is single-valued for almost all inputs and always for sufficiently small step sizes $s_x$ and $s_\lambda$ such that $s_xs_\lambda<1$.
We will assume in this section that $\tau=1$, since a different $\tau$ simply scales the step sizes.

Consider the proximal operator of the variable-coefficient $\ell_1$ norm function $g(\mathbf{x}, \boldsymbol\lambda)=\sum_{p=1}^P \lambda_p |x_p|$. 
Since this function decomposes into additive functions of each $(\lambda_p$, $x_p)$ individually, its proximal operator acts on each block independently of the others. 
Therefore, for the remainder of this section, we consider a single block $(\lambda, x)$, dropping the index $p$, and consider the proximal operator of the 2-dimensional function $g(x,\lambda) = \lambda |x|$:
\begin{flalign}\tag{P1}\label{eq:prox_prob}
     \mathrm{prox}^{s_\lambda,s_x}_g (x_0, \lambda_0) = \underset{x\in\mathbb{R},\lambda>0}{\mathrm{argmin}} \,\,  \lambda |x|+\frac{(x-x_0)^2}{2s_x} + \frac{(\lambda-\lambda_0)^2}{2s_\lambda} \,\, . & \hspace{5em}
\end{flalign}

\begin{lemma}\label{lem:marg}
    The marginal cost of \ref{eq:prox_prob} with respect to $\lambda$ (i.e. with $x$ profiled out) is the following piecewise quadratic expression:
\begin{equation}
     \underset{\lambda>0}{\mathrm{argmin}}  \begin{cases}
        \frac{1}{2}(\frac{1}{s_\lambda} - s_x)\lambda^2+(|x_0|-\frac{\lambda_0}{s_\lambda})\lambda + \frac{\lambda_0^2}{2s_\lambda} & \lambda < \frac{|x_0|}{s_x} \\
        \frac{(\lambda-\lambda_0)^2}{2s_\lambda}+\frac{x_0^2}{2s_x} & \lambda \geq \frac{|x_0|}{s_x}\,\, , \\
     \end{cases}
\end{equation}
where the changepoint $\lambda=\frac{|x_0|}{s_x}$ is the point where $\lambda$ is just large enough to push $x$ to zero.
\end{lemma}
\begin{proof}
    Convert to nested optimization and exploit the known solution for fixed $\lambda$ given by the soft thresholding operator.
\end{proof}
The quadratic polynomial in the interval $[\frac{|x_0|}{s_x},\infty)$ is always convex. When $s_\lambda s_x<1$, the quadratic polynomial in the other interval is convex as is the overall expression. But when $s_\lambda s_x>1$, the coefficient of the quadratic term is negative, and that polynomial is concave, yielding a nonconvex piecewise function (see Figure \ref{fig:nonconv}, center and right). 

We are now prepared to develop the proximal operator.


\begin{theorem}
    The optimizing $\lambda$ for the proximal program \ref{eq:prox_prob} is given by, when $s_x s_\lambda<1$:
    \begin{equation}\label{eq:prox1}
        \lambda^* =\begin{cases} 
          \lambda_0 & \lambda_0 \geq \frac{|x_0|}{s_x} \\
          \frac{(\lambda_0-s_\lambda|x_0|)^+}{1-s_\lambda s_x} & o.w. \,\,\,\, ,
       \end{cases} 
    \end{equation}
    and, when $s_x s_\lambda\geq1$, by $\lambda^* = \mathbbm{1}_{\big[\frac{\lambda_0}{\sqrt{s_\lambda}} > \frac{|x_0|}{\sqrt{s_x}}\big]} \lambda_0$  (here $\mathbbm{1}$ denotes the indicator function).
In either case $x^* = (|x_0|-s_x\lambda^*)^+\mathrm{sgn}(x_0)$.
\end{theorem}
\begin{proof}
    We need only compare the optima of the quadratic functions of Lemma \ref{lem:marg}.
\end{proof}


Due to the nonconvexity of the proximal cost, this proximal program may have two global optima. 
Thus the proximal operator is discontinuous and multi-valued at the discontinuity, as visualized in the top right of Figure \ref{fig:prox_act}.



\begin{remark}\label{rm:dualsparse}
    When $s_x s_\lambda<0$ and $\lambda_0 < s_\lambda |x_0|$ or when $s_x s_\lambda>1$ and $\frac{\lambda_0}{\sqrt{s_\lambda}}>\frac{|x_0|}{\sqrt{s_x}}$, the solution to the proximal problem gives $\lambda = 0$, which would lead to no shrinkage on $\beta$.
\end{remark}
Remark \ref{rm:dualsparse} is interesting, as it implies that it is possible to develop a procedure with ``dual sparsity": on the regression coefficient, when appropriate, or on the penalty coefficient.
However, in our application of this operator to the Laplace penalty, the $\lambda$ normalization term gives this a density of zero, precluding that point being a penalized maximizer.
We look forward to examining other models which allow for zero penalties.

\begin{figure}[h]
    \centering
    {\tiny\textbf{Fixed $s_xs_\lambda$:}}

    \includegraphics[scale=0.3]{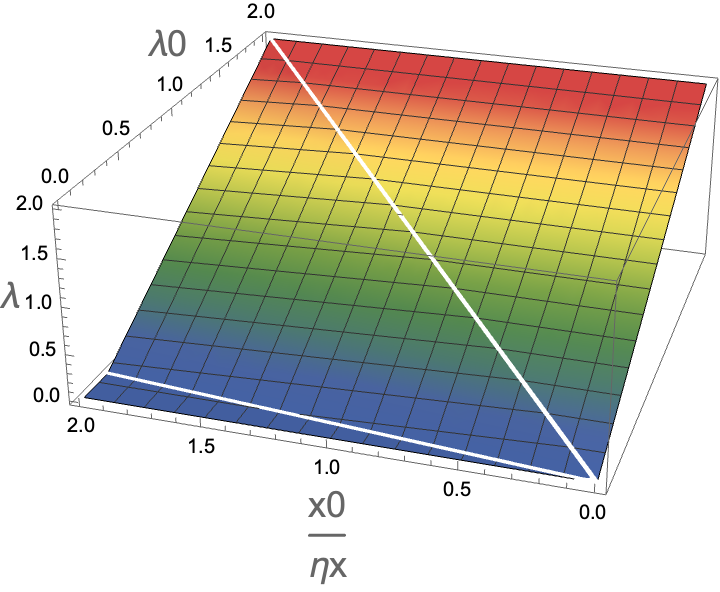}
    \includegraphics[scale=0.3]{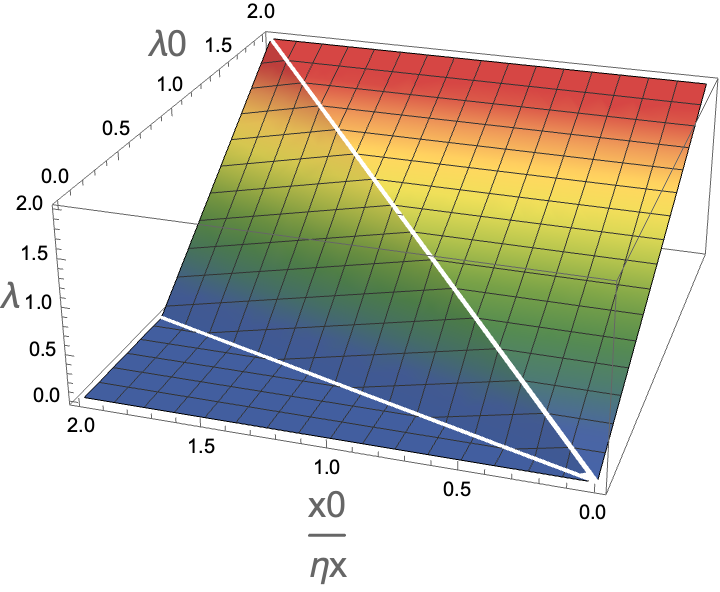}
    \includegraphics[scale=0.3]{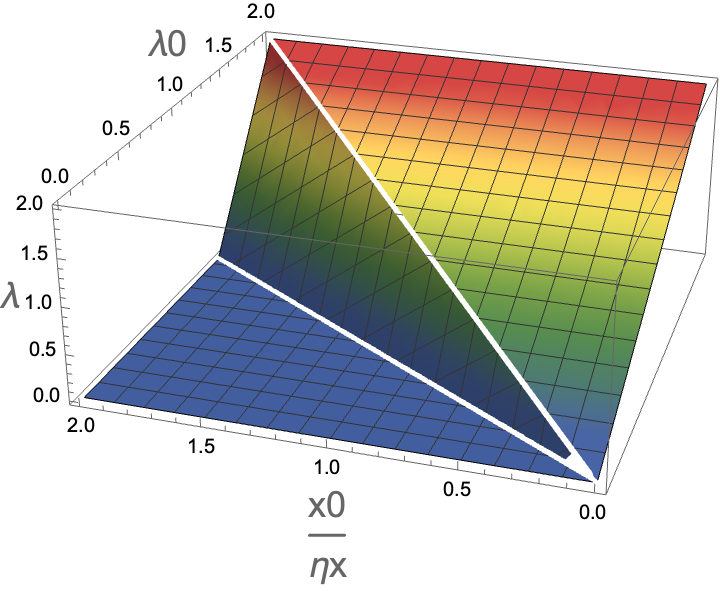}
    \includegraphics[scale=0.3]{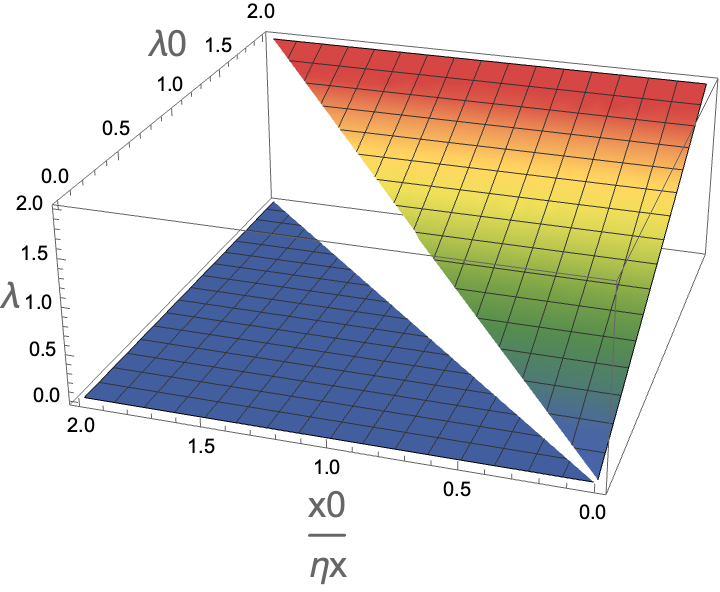}
    
    {\tiny\textbf{Fixed $\frac{|x_0|}{s_x}$:}}

    \includegraphics[scale=0.3]{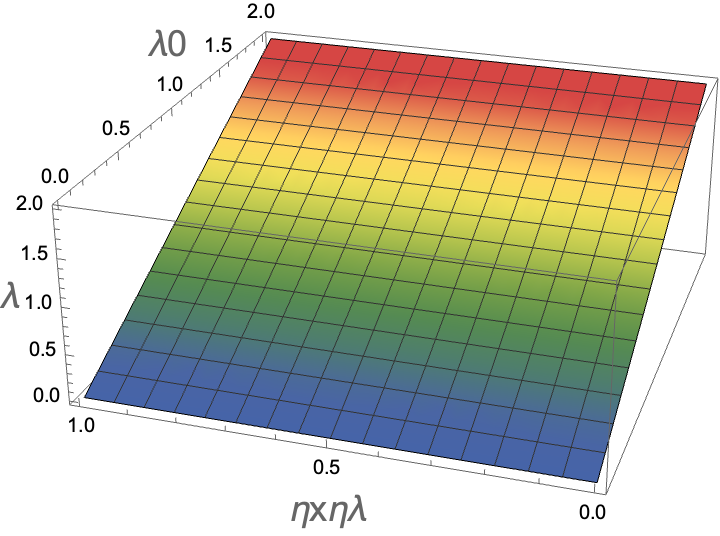}
    \includegraphics[scale=0.3]{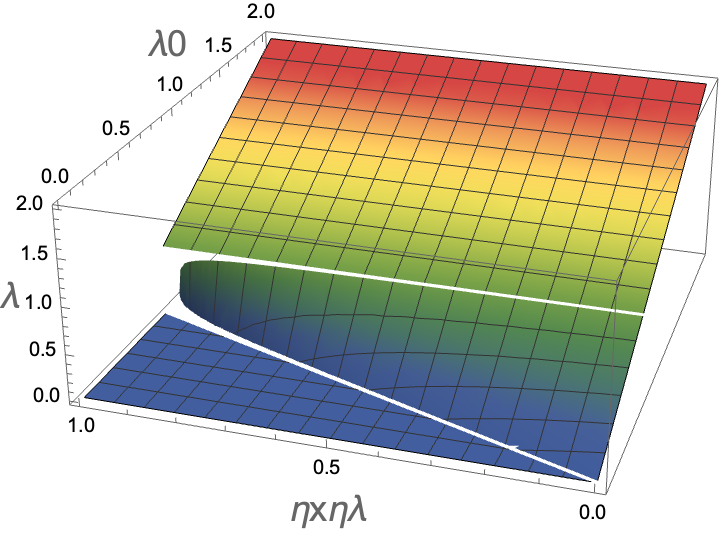}
    \includegraphics[scale=0.3]{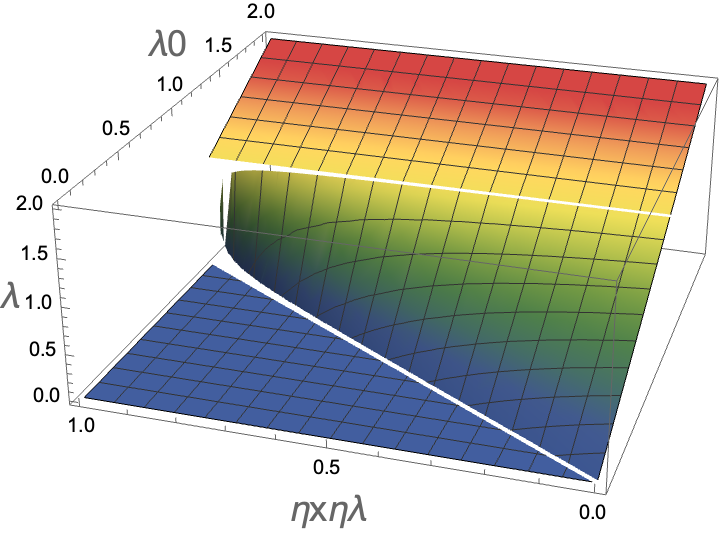}
    \includegraphics[scale=0.3]{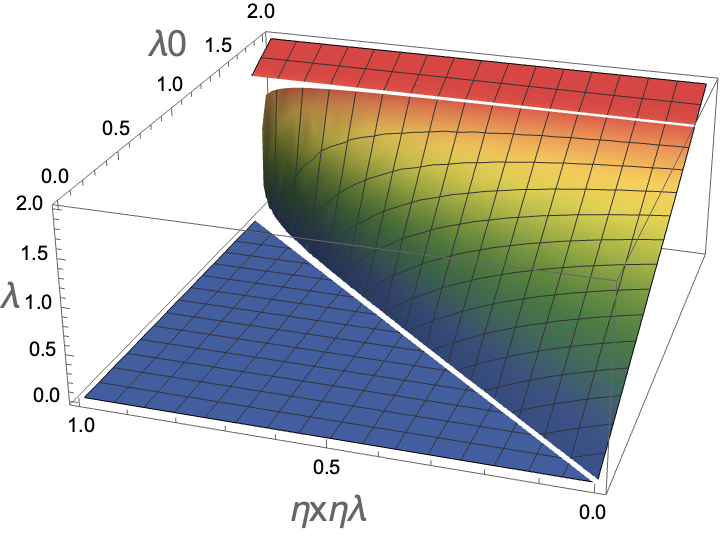}
    
    {\tiny\textbf{Fixed $\lambda_0$:}}

    \includegraphics[scale=0.3]{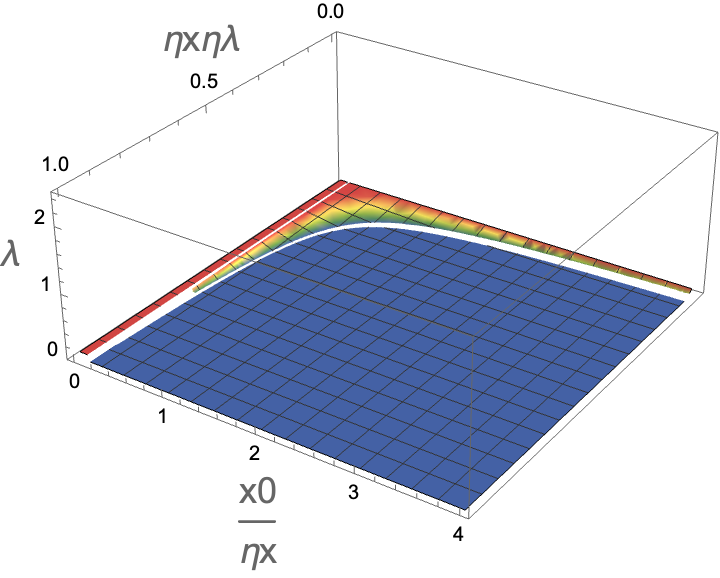}
    \includegraphics[scale=0.3]{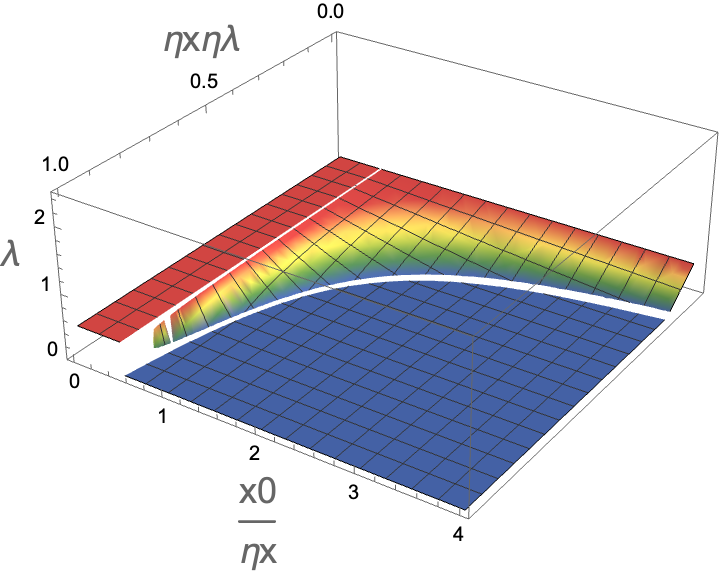}
    \includegraphics[scale=0.3]{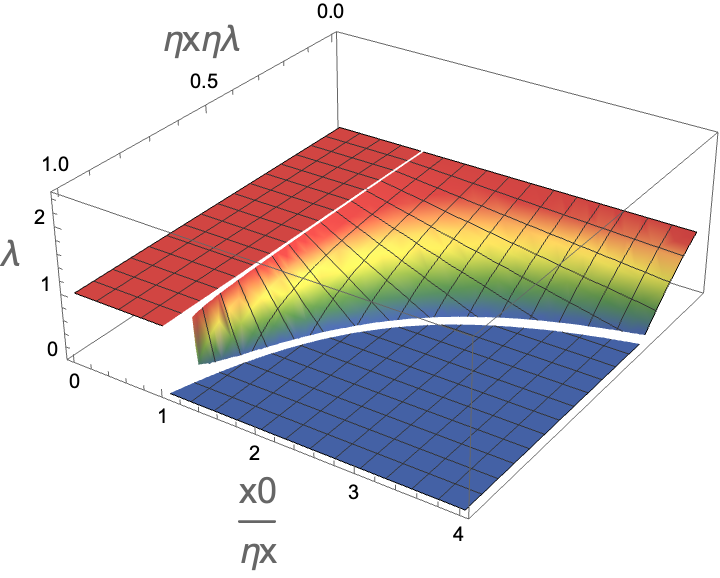}
    \includegraphics[scale=0.3]{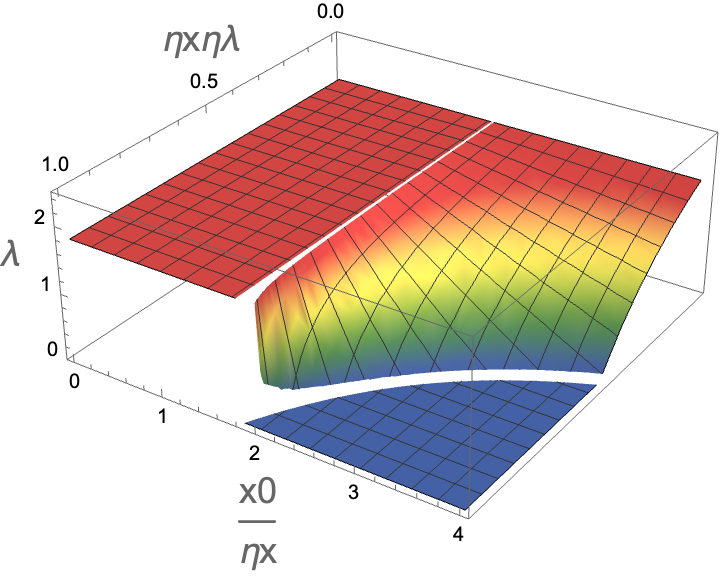}

    \caption{
    \textbf{The Action of the Proximal Operator:}
    Plots of the reduced proximal operator (Eq \ref{eq:reduced_prox}) considering two variables at a time and fixing the third.
    \textit{Top:} for various fixed $b:=s_x s_\lambda<1$ and with $\lambda_0,\frac{|x_0|}{s_x} \in (0,2)$.
    Values $b=s_x s_\lambda\in\{0.1,0.35,0.65,0.99\}$ are shown left to right.
    \textit{Mid:} $a=\frac{|x_0|}{s_x}\in\{0,0.5,1,1.8\}$. 
    \textit{Bottom:} $\lambda_0\in\{0.2,0.75,1.25,1.75\}$.
    }
    \label{fig:prox_act}
\end{figure}

This proximal operator has been conceptualized as a mapping of $(\lambda_0,x_0)\to(\lambda^*,x^*)$ parameterized by $s_x$ and $s_\lambda$, but to visualize it we will briefly study it as a function of these four quantities mapping to an optimizing $\lambda^*$.

\begin{remark}
Since when $s_x s_\lambda>1$ the $\lambda^*$ is either $0$ or $\lambda_0$, we will focus on the case where $s_x s_\lambda<1$. 
Then, let $a:=\frac{|x_0|}{s_x}$ and $b:=s_xs_\lambda$ yielding a function of just three variables:
\begin{equation}\label{eq:reduced_prox}
    \lambda(\lambda_0,a,b) = \begin{cases}
        \lambda_0 & \lambda_0 \geq a \\
        \frac{(\lambda_0-ab)^+}{1-b} & o.w. \,\,\,\,\, .
    \end{cases}
\end{equation}
\end{remark}
This function is visualized in Figure \ref{fig:prox_act}. 
For given step size product $b$ the function is stepwise linear, and converges to the identity mapping with respect to its $\lambda_0$ input as $b\to 0$. 
As $b\to 1$, the mapping becomes more and more steep for $\lambda_0\in(ab,a)$, gradually converging to the discontinuous mapping $\lambda(\lambda_0,a,b)\to \lambda_0 \mathbbm{1}_{\lambda_0 > a}$ .
The proximal operator is multi-valued at $\lambda_0=a$ when $b=1$.

\noindent\textbf{Comparison to Other Penalties in Statistics:} Existing nonconvex alternatives to our proposed biconvex penalty in the form of the MCP, SCAD and Bridge penalties all come with an extra hyperparameter controlling how close the penalty gets to approximating $||.||_0$ directly.
The proposed approach retains only the single global penalty strength parameter, which in practice means it is possible to simply try all pertinent parameter values via a sequence of warm starts.
However, it does require specification of a prior on $\lambda$, though we have found the standard Half-Cauchy to be generally sufficient.
\citet{candes2008enhancing} and \citet{fan2014strong} propose to set $\lambda_p$ in order to locally approximate a prespecified nonconvex penalty, and so also updates $\lambda_p$ each iteration.

\subsection{Deploying the $\mathrm{prox}$ operator with VISTA}

We can deploy the proximal operator of Equation $\ref{eq:prox1}$ as part of a proximal gradient method:
\begin{align}
    & \tilde{\boldsymbol\beta}^{t+1} \gets \boldsymbol\beta^t - s_{\boldsymbol\beta} \nabla_{\boldsymbol\beta} \mathcal{L}(\boldsymbol\beta^t, \boldsymbol\lambda^t) \\
    & \tilde{\boldsymbol\lambda}^{t+1} \gets \boldsymbol\lambda^t - s_{\boldsymbol\lambda} \nabla_{\boldsymbol\lambda} [\mathcal{L}(\boldsymbol\beta^t, \boldsymbol\lambda^t)-\log\boldsymbol\lambda^t-\log P_{\boldsymbol\lambda}(\boldsymbol\lambda^t)] \\
    & \boldsymbol\beta^{t+1},\boldsymbol\lambda^{t+1} \gets \mathrm{prox}_{g}^{s_{\boldsymbol\lambda},s_{\boldsymbol\beta}}(\tilde{\boldsymbol\beta}^{t+1},\tilde{\boldsymbol\lambda}^{t+1})
\end{align}
In practice, success with any gradient descent method relies on choosing good step sizes, preconditioners, and acceleration, as we detail in Supplementary 
\if\arxiv 1
2.
\else
\ref{sec:ap_optim}.
\fi


\subsection{Basic Theoretical Properties for Penalized Likelihood}

In this section, we will assume that $\boldsymbol\lambda$ is excluded from the misfit term $\mathcal{L}(\boldsymbol\beta)$ such that it only appears in the penalty and, eventually, that this misfit term $\mathcal{L}(\bb)$ is given by a negative likelihood $-L(\y;\bb)$. In this case, we can rewrite our penalty as such (we will assume all parameters are penalized for ease of discussion):
\begin{align}
    & \underset{\boldsymbol\beta,\boldsymbol\lambda>\mathbf{0}}{\min}\,\, -L(\boldsymbol\beta) + \sum_{p=1}^P \big[\tau\lambda_p |\beta_p| - \log\lambda_p\big] + \sum_{p=1}^P -\log p_\lambda (\lambda_p) \\
    & \iff \underset{\boldsymbol\beta}{\min}\,\, -L(\boldsymbol\beta) + \sum_{p=1}^P \underset{\lambda_p>0}{\min} \big[\tau\lambda_p |\beta_p| - \log\lambda_p  -\log p_\lambda (\lambda_p) \big]  \label{eq:penlik}
\end{align}
Therefore, we may profile over $\lambda$ to develop a penalty $g_{\tau}(|\beta|) = \underset{\lambda>\mathbf{0}}{\min} \, \tau\lambda |\beta| + \log\lambda  + \rho(\lambda)$, where $\rho(\lambda) = -\log p_\lambda (\lambda)$.
We begin with some basic properties of this penalty.
\begin{lemma}
	The following hold, where $\lambda^*$ denotes the optimizing $\lambda$, and is formally a function of $\tau$ and $|\beta|$:
 \begin{multicols}{2}
	\begin{enumerate}
		\item $\lambda^* = \frac{1}{\tau|\beta| + \rho'(\lambda^*)}$.
		\item $\frac{\partial \lambda^*}{\partial |\beta|} = -\frac{\tau}{\frac{1}{\lambda^{*2}} + \rho''(\lambda)}$.
		\item $g_{\tau_n}'(|\beta|) =  \tau \lambda^*$.
		\item $g_{\tau_n}''(|\beta|) =  -\frac{\tau_n^2}{(\tau_n+\rho'(\lambda^*))^2+\rho''(\lambda^*)}$.
	\end{enumerate}
 \end{multicols}
\end{lemma}
\begin{proof}
    These follow from implicit differentiation on first order optimality conditions.
\end{proof}
This allows us to quantify the behavior of this penalty as follows:
\begin{theorem}
    Assume that the logarithmic derivative of the hyperprior density on $\lambda$ is bounded ($|\rho'(\lambda)|< M_1 \,\, \forall \lambda\geq0$) and that the density is decreasing on $(0,\infty)$. Then:
    \begin{enumerate}
        \item $g'_{\tau}(|\beta|)\approx\frac{1}{|\beta|}$ for large $\beta$.
        \item The minimum of $|\beta|+g'_{\tau}(|\beta|)$ is achieved at $\beta=0$ with value $\lambda_a\tau$.
    \end{enumerate}
\end{theorem}
\begin{proof}
    Follows from Lemma 2.
\end{proof}

\begin{remark}
    $g'_{\tau}(|\beta|)\approx\frac{1}{|\beta|}$ approximates the gradient of the adaptive lasso \citep{Zou2006} procedure with optimal weights with hyperparameter $\gamma=1$ as well as the iteration described by \citet{candes2008enhancing}.
\end{remark}
\begin{remark}
    \citet{Fan2001} describe three desirable properties of nonconcave penalties: first, that they have bias decreasing quickly in nonzero parameter size, second, that they induce sparsity, and third, that they be continuous in the data. The shrinking gradient size with parameter norm is sufficient to ensure unbiasedness for large parameters. The fact that the minimum of $|\beta|+g'_{\tau}(|\beta|)$ is strictly positive ensures sparsity, while the fact that the minimum occurs at zero ensures continuity. This latter condition is not satisfied by, for example, the bridge penalty.
\end{remark}
\begin{remark}
Bounded logarithmic derivatives are satisfied by the densities of, for example, the Cauchy and Exponential distributions, but not the Gaussian distribution.
\end{remark}

We next consider the asymptotic distribution of the penalized likelihood estimator. In particular, we demonstrate that there exists a local minimizer of the penalized loss which satisfies the oracle property of \citet{Fan2001}.
We assume that the first $r$ entries of the true parameter vector $\bb_0$ are nonzero, and the rest 0, such that $\bb_0={\tiny\begin{pmatrix} \bb_{10}\\ \bb_{20}=\mathbf{0}\end{pmatrix}}$.
\begin{theorem}
    Let $\tau_n=n\tau_0$ for $\tau_0>0$, and further assume that $|\rho''(|\lambda|)|<M_2$ (bounded second logarithmic derivative). Then, under the standard regularity conditions on the likelihood enumerated in the supplementary material, there is a local minimum of $\ref{eq:penlik}$ that satisfies the following:
    \begin{enumerate}
        \item $\hat{\bb}_2=\mathbf{0}$ with probability approaching 1 as $n\to\infty$.
        \item $\hat{\bb}_1$ is asymptotically normal with covariance given approximately by $\frac{1}{n}I^{-1}(\bb_{10})$, the Fisher information matrix considering only active variables.
    \end{enumerate}
\end{theorem}
\begin{proof}
See Supplementary 
\if\arxiv 1
1.
\else
\ref{sec:app_proofs}.
\fi
\end{proof}
The oracle property tells us that this estimator has the same asymptotic distribution as that estimator with truly zero $\beta_j$ clamped to zero.

\section{The Sparse Bayesian Lasso}\label{sec:linear}

Having developed a general-purpose optimization algorithm and developed basic properties in the penalized likelihood case, we pivot to develping a Bayesian Lasso procedure which provides full uncertainty quantification and penalty coefficient adaptation.
We achieve this by deploying the VISTA procedure of the previous section on a novel nonsmooth Variational Bayesian methodology which we call the Sparse Bayesian Lasso.

\subsection{Variational Inference and Nonsmooth Penalties}\label{sec:vb}



Variational Bayes searches for a distribution over an unknown parameter vector $\boldsymbol\theta$ that is 1) analytically tractable, and 2) sufficiently close to the posterior such that what each has to say about the quantities of interest are approximately the same.
It does so by defining an optimization problem $\underset{Q_\theta \in \mathcal{Q}}{\mathrm{argmin}}\, d(Q_\theta,P_{\theta|\y})$
where $P_{\theta|\y}$ is our true posterior probability with density $p$ and $Q_{\theta}$ is a candidate variational distribution with density $q$ from a space of possible distributions $\mathcal{Q}$ parameterized by \textit{variational parameters}, and $d$ is some measure of dissimilarity between distributions.
In this article, for each parameter $\theta$, the variational distribution is defined by a location parameter $\eta_\theta$ and a scale parameter $\nu_\theta$.
We will impose variational independence as is common: $q_{\boldsymbol\eta,\boldsymbol\nu}(\boldsymbol{\theta}) = \prod_{m=1}^M q_{\eta_m,\nu_m}(\theta_m)$.

The most common $d$ is the KL divergence from variational to posterior:
\begin{equation}
    d(Q,P) = \mathrm{KL}(Q_{\theta}||P_{\theta|\mathbf{y}}) = \mathbb{E}_{\boldsymbol\theta\sim Q}[\log\bigg(\frac{q(\boldsymbol{\theta})}{p(\boldsymbol{\theta}|\mathbf{y})}\bigg)]\, 
    = -\mathbb{E}_{\theta\sim Q}[\log \mathcal{L}(\mathbf{y}|\boldsymbol\theta)] + \mathrm{KL}(Q_{\theta}||P_{\theta})  \, .
\end{equation}
We see this $d(Q_\theta,P_{\theta|\y})$ may be viewed as the negative expected log likelihood (which we can view as a model misfit term) penalized by the KL divergence between the variational and prior distributions (which we can view as a model complexity term).
Since $Q$ was chosen to be simple, we can use a Monte Carlo sample to estimate the expected likelihood.

For example, for regression coefficients $\beta_p$, we might specify Laplace distributions both for variational and prior distributions: $\beta_p\overset{Q}{\sim}\mathrm{L}(\eta_{\beta_p},\nu_{\beta_p})$ and $\beta_p\overset{P}{\sim}\mathrm{L}(0,\frac{1}{\lambda_p\tau})$, leading to the following penalty function on our variational parameters $\eta_{\beta_p}$ and $\nu_{\beta_p}$:
\begin{equation}\label{eq:laplace_kl}
    g_{\mathrm{KL}}(\eta_{\beta_p},\nu_{\beta_p},\lambda_p)=\mathrm{KL}(\mathrm{L}(\eta_{\beta_p}, \nu_{\beta_p})|| \mathrm{L}(0, \frac{1}{\lambda\tau})) = \tau\lambda (\nu_{\beta_p} e^{-\frac{|\eta_{\beta_p}|}{\nu_{\beta_p}}}+|\eta_{\beta_p}|) - \log(\nu_{\beta_p}) - \log\lambda \, . 
\end{equation}

One might be forgiven for thinking that Variational Bayes, a procedure built on optimization, would be the approach that allows us to combine the generality afforded by Bayesian inference with the sparsity of penalized likelihood procedures.
Unfortunately, this is not the case.
The proof of this is straightforward, but we have not seen it explicitly mentioned in the academic literature.

\begin{theorem}
    Penalty functions defined by KL-cost Variational Bayesian procedures with continuous prior and variational distributions do not induce sparsity.
\end{theorem}
\begin{proof}
    As \citet[Section 2]{Fan2001} discuss in the context of penalized likelihood, any penalty function which induces sparsity and is continuous in the data must have a singularity at the origin. 
    Since Variational Bayesian penalties are defined by a prior-variational KL divergence, and this in turn is defined by an integral, an application of the fundamental theorem of calculus is sufficient to reach the desired conclusion.
\end{proof}


The penalty given in Equation \ref{eq:laplace_kl} is smooth despite the two absolute values because they cancel at zero as visualized in Figure $\ref{fig:huber}$ and demonstrated algebraically by \citet{Meyer2021}, who studied this function because of its similarity to the Huber loss function.

\subsection{A Nonsmooth Approximation}

\citet{owen2007robust}, in the context of penalized M estimation, advocates to reintroduce nonsmoothness by exchanging the quadratic and linear parts of the penalty.
We propose to approximate $\nu_{\beta_p} e^{-\frac{|\eta_{\beta_p}|}{\nu_{\beta_p}}}+|\eta_{\beta_p}|$ by $\nu_{\beta_p}+|\eta_{\beta_p}|$ since it is a good approximation for small $\nu_{\beta_p}$:
\begin{equation}\label{eq:laplace_ns}
     g_{\mathrm{NS}}(\eta_{\beta_p},\nu_{\beta_p},\lambda_p) := \tau\lambda_p (\nu_{\beta_p} + |\eta_{\beta_p}|) - \log(\nu_{\beta_p}) - \log\lambda_p - 1 \, .
\end{equation}
Figure \ref{fig:huber} show that $g_{\mathrm{NS}}$ retains the qualitative properties of the KL penalty, and $g_{\mathrm{NS}}(\eta_{\beta_p},\nu_{\beta_p})-1$ bounds $g_{\mathrm{KL}}(\eta_{\beta_p},\nu_{\beta_p})$ from below.
Notice that the term $\log\eta_{\beta_p}$ serves as a logarithmic barrier preventing $\eta_{\beta_p}\to0$. 
Therefore, the sparsity induced by this penalty is qualitatively different from that induced by the Horseshoe, or Spike-Slab, which compress the entire posterior distribution to zero.
In contrast, this proposed penalty sets only the variational mode exactly to zero while allowing for nonzero, and sometimes significant, variational variance.

\begin{figure}
	\centering
	\raisebox{-0.5\height}{\includegraphics[scale=0.8]{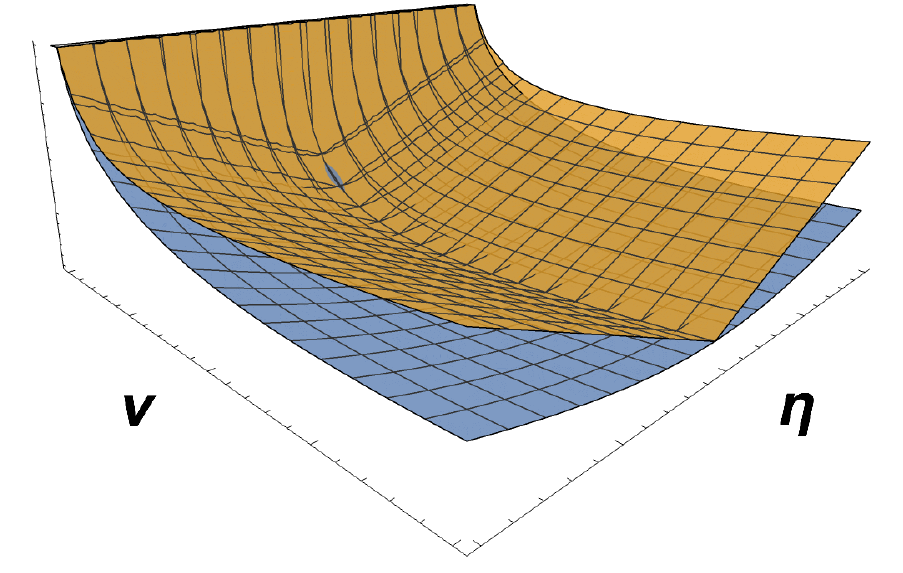}}
	\raisebox{-0.5\height}{\includegraphics[scale=0.45]{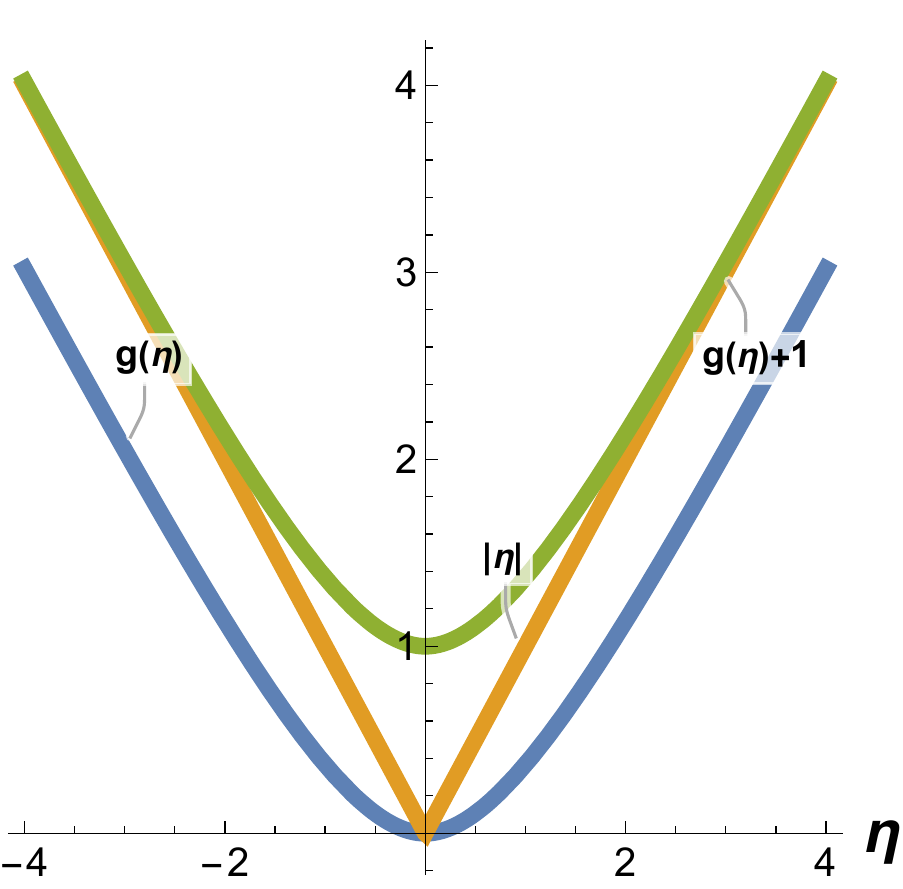}}
	\,\,
	\raisebox{-0.5\height}{\includegraphics[scale=0.45]{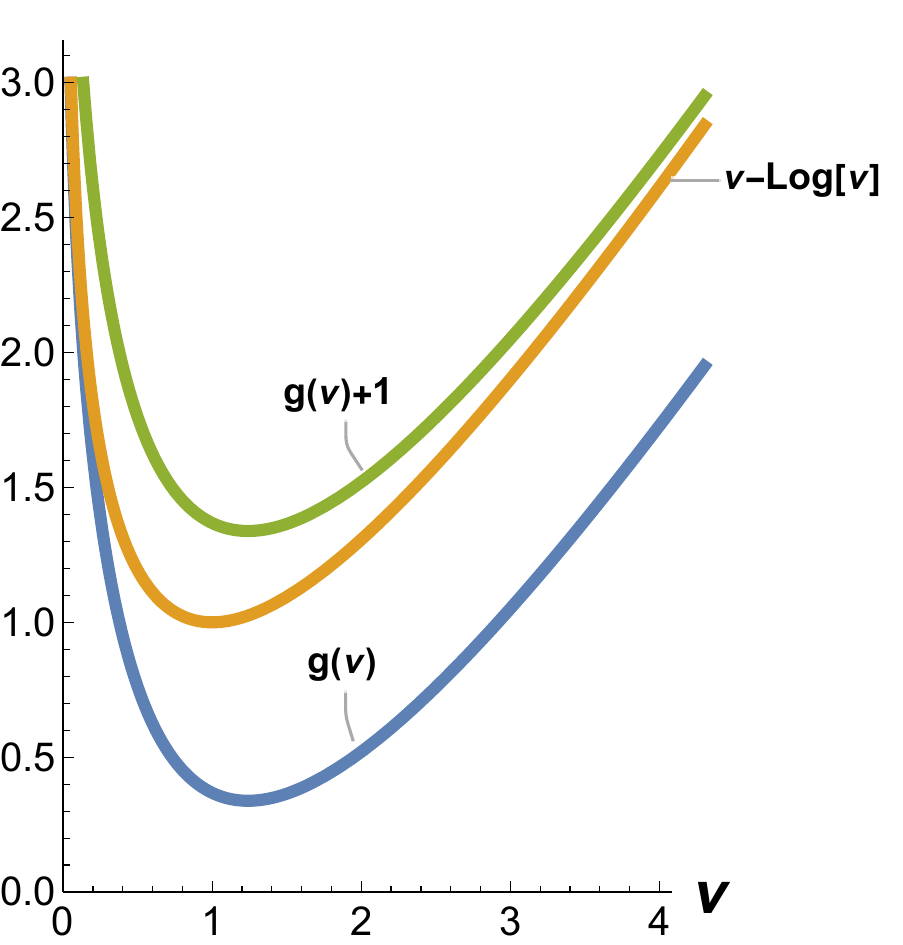}}
	\caption{\textbf{Nonsmooth Approximation} 
	\textit{Left:} Bivariate KL penalty (blue) and nonsmooth approximation (orange).
	The marginal penalties, ($g$ gives KL) for
	\textit{Center:} $\eta$, 
	\textit{Right:} $\nu$. 
	}
	\label{fig:huber}
\end{figure}

In the case of a linear regression with Gaussian error, the cost is closed form whenever the first two moments of $Q_{\boldsymbol\beta}$ and the inverse and log expectation of $Q_{\sigma^2}$ are tractable (Supplementary
\if\arxiv 1
3).
\else
\ref{sec:app_lineargaus}).
\fi
We cannot expect this to be the case in general, and must estimate the expected log likelihood term $-\mathbb{E}_{\boldsymbol\theta\sim Q_{\boldsymbol\theta}}[\log \mathcal{L}(\mathbf{y}|\boldsymbol\theta)]$. 
Most common in ``blackbox" VB \citep{Ranganath2014} is to optimize this with stochastic optimization over a Monte Carlo sample.
We instead use the Sample Average Approximation \citep[SAA]{Robinson1996}, which converts the optimization to a deterministic one by employing the same Monte Carlo draw throughout the optimization procedure, and use antithetic sampling for variance reduction \citep[Chapter 8.2]{Owen2013}.
Collecting those parameters which are not desired to be sparse into the vector $\boldsymbol\theta$ with associated variational parameters $\eta_\theta,\nu_\theta$, our cost is given as:
\begin{equation}
    \underset{\boldsymbol\eta}{\min}\,\, \frac{1}{B}\underset{\boldsymbol\theta_b\sim Q_{\boldsymbol\theta},\boldsymbol\beta_b\sim Q_{\boldsymbol\beta}}{\sum} \mathcal{L}(\boldsymbol\theta_b,\boldsymbol\beta_b) + \sum_{p=1}^P g_{\mathrm{NS}}(\eta_{\beta_{b,p}},\nu_{\beta_{b,p}},\lambda_p) + \mathrm{KL}(Q_{\boldsymbol\beta,\boldsymbol\theta}||P_{\boldsymbol\beta,\boldsymbol\theta})
\end{equation}
Note that to compute gradients we must differentiate through the sampling procedure; this is generally implemented in automatic differentiation packages like \texttt{tensorflow} and \texttt{pytorch}, though quality of the gradients vary.
When the distribution is a location-scale family, this differentiation is particularly straightforward \citep{Kingma2014}. 
With the gradient of the smooth part computed, we can apply the VISTA algorithm to yield sparse learning with any smooth likelihood.
Since $\lambda_p$ is of secondary interest, we estimate it via maximization rather than specifying a distribution for it.

\subsection{Building Coefficient Trajectories by Varying \texorpdfstring{$\tau$}{tau}}\label{sec:qual}

\begin{figure}
    \centering
    \includegraphics[scale=0.52,trim={0 6.5cm 0 0},clip]{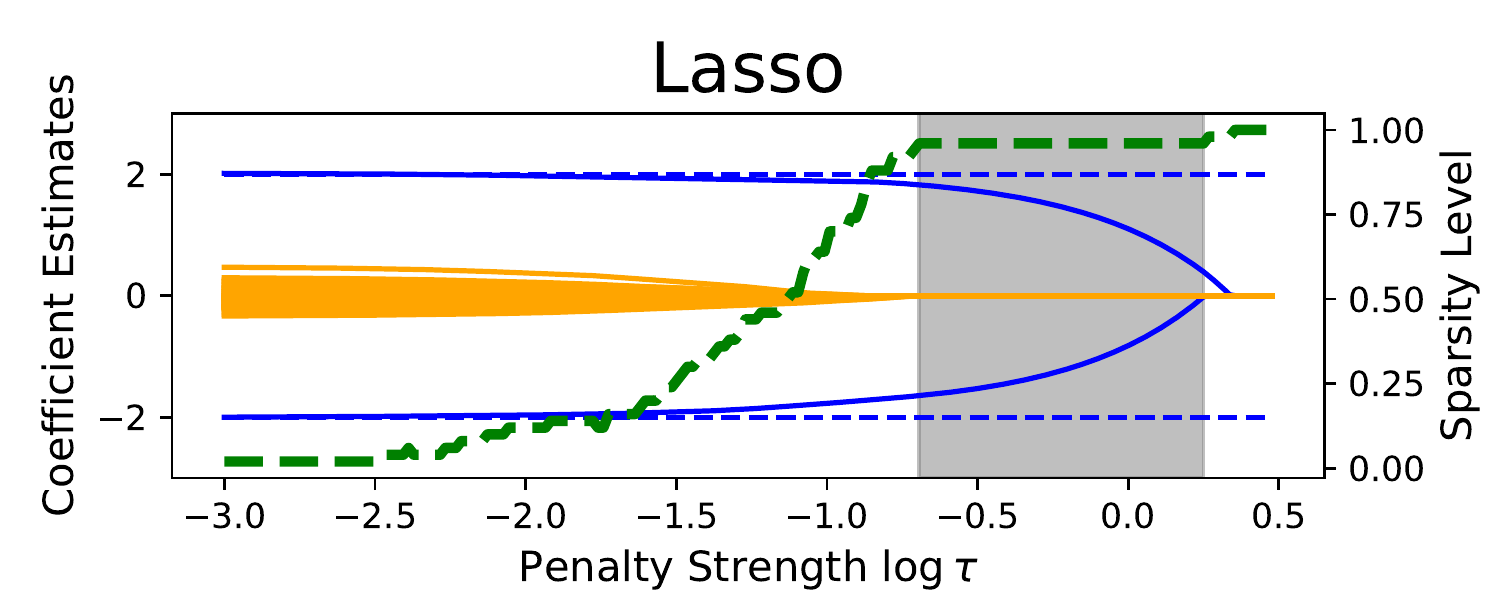}
    \includegraphics[scale=0.52,trim={0 0 0 6.5cm},clip]{example_traj.pdf}
    \caption{\textbf{Qualitative Behavior:} Solid lines indicate coefficient estimates (y-axis) as function of regularization strength for Lasso (left) SBL (right). Dotted blue lines give true nonzero coefficient values (-2,2), and solid blue lines give estimates. Orange lines give estimates of coefficients which are 0. Dotted green line (second y-axis) gives percent of estimated coefficients estimated as 0 (sparsity).}
    \label{fig:example_traj}
\end{figure}

A practically useful aspect of Lasso regression is the ability to efficiently form trajectories of parameter estimates while varying the penalty strength $\tau$ by warm-starting optimizations.
We also deploy VISTA estimating the SBL with Gaussian error structure and a Half-Cauchy prior, illustrated in Figure \ref{fig:example_traj}, evaluate a trajectory for varying $\tau$, and then compare to the Lasso on a simple simulated dataset with a true Gaussian-linear model in 50 dimensions. 
Though it successfully thresholds the active parameters last, there is no regularization strength which allows Lasso to identify the correct model without bias. 
In the SBL, by contrast, the variational mean of the active parameters are near their true values for the whole range of penalty values until they are thresholded out, closer to the behavior of nonconvex penalties such as MCP and SCAD.
Unlike penalized likelihood procedures, however, we can build trajectories of entire variational distributions, with credible intervals pictured here.


\subsection{Empirical Frequentist Properties}

\begin{table}
\tiny
\centering
\begin{tabular}{llrrrrr}
\toprule
Likelihood & Method & $\beta$-FNR &  $\beta$-FPR &  $\beta$-coverage &  $\sigma^2$-coverage &   Time (s) \\
\midrule
bernoulli & VISTA-MAP &      0.03 &      0.00 &        - &       - &   4.43 \\
        & VISTA-SBL &      0.00 &      0.00 &       0.91 &      1.00 &  31.58 \\
        & Horseshoe &      0.00 &      0.04 &       0.93 &      1.00 & 400.09 \\
cauchy & SSLasso &      0.88 &      0.00 &        - &       - &   0.02 \\
        & VISTA-MAP &      0.30 &      0.00 &        - &       - &   6.43 \\
        & VISTA-SBL &      0.00 &      0.03 &       0.94 &      0.97 &  40.86 \\
        & Horseshoe &      0.00 &      0.00 &       0.94 &      0.87 & 458.67 \\
nb              & VISTA-MAP &      0.00 &      0.02 &        - &       - &   6.31 \\
        & VISTA-SBL &      0.00 &      0.00 &       0.95 &      1.00 &  28.77 \\
        & Horseshoe &      0.00 &      0.00 &       0.94 &      1.00 & 962.86 \\
 normal & SSLasso &      0.00 &      0.00 &        - &       - &   0.02 \\
        & VISTA-MAP &      0.00 &      0.02 &        - &       - &   7.12 \\
        & VISTA-SBL &      0.00 &      0.00 &       0.94 &      1.00 &  34.93 \\
        & Horseshoe &      0.00 &      0.00 &       0.95 &      1.00 &  39.10 \\
poisson         & VISTA-MAP &      0.00 &      0.00 &        - &       - &  18.01 \\
        & VISTA-SBL &      0.00 &      0.00 &       0.80 &      1.00 &  46.77 \\
        & Horseshoe &      0.00 &      0.00 &       0.98 &      1.00 & 394.34 \\
\bottomrule
\end{tabular}

\caption{FNR is False Negative Rate, FPR is False Positive Rater, COVR is 95\% empirical interval coverage, time is elapsed real time in seconds for given $\tau$.}\label{tab:sim}
\end{table}

In this section, we examine the empirical performance of our proposed method on selected GLMs.
We use the mean function described in our working example.
We deploy the VISTA algorithm, implemented in \texttt{tensorflow\_probability}, on for both MAP and VB inference with Bernoulli, Negative Binomial, Normal, and Poisson likelihoods using the canonical link functions and with the Cauchy likelihood with an identity link.
We tuned sparsity level $\tau$ by hand which is what we advocate doing in practice by examining the entire trajectory produced by varying $\tau$ (Supplementary
\if\arxiv 1
5), fixing it across Monte Carlo replicates.
\else
\ref{sec:app_imp}), fixing it across Monte Carlo replicates.
\fi

We use a Monte Carlo sample size of 40.
As a comparator, we show a \texttt{JAGS} implementation of the Horseshoe prior, run for 1,000 iterations after adaptation on two chains, which should represent a lower bound on the amount of computation required for inference in practice.
Additionally, we compare to the Spike-Slab Lasso using the R package \texttt{SSLASSO} when the link function is the identity, namely on the Gaussian and Cauchy cases.
The Spike-Slab Lasso penalty could of course be combined with a non-Gaussian likelihood, but solving the marginal problem for coordinate descent will no longer in general be in closed form and is not implemented in \texttt{SSLASSO}.
In contrast, the proximal operator is the same for any differentiable likelihood.

As Table \ref{tab:sim} shows, SBL is able to maintain high interval coverage using generally about an order of magnitude the elapsed real time of the Horseshoe prior.
The performance gap is smaller on the Normal likelihood where the Gibbs sampler enjoys conjugacy.
On the Negative Binomial likelihood, which is the likelihood we use in the case study of Section \ref{sec:results}, we observe about a $30\times$ speedup.
The SSLasso is fast but not yet applicable to non-Gaussian likelihoods and unable provide UQ.

\section{Modeling Migration in Iraq}  \label{sec:application}

We begin this section with an overview of the classical Gravity Model.
We then pivot to discussing the data set we collected, and  finally present the results of our analysis.

\subsection{ISIL and the 2013-2017 Iraqi Civil War}

In late 2013, tensions between Sunni and Shia Muslims in the Anbar province of Iraq boiled over into violence between the Shia government and its allies against the Islamic State.
The violence led to significant internal displacement which is the subject of our case study.
In the absence of traditional migration variables, we consider the use of social media for capturing indirect indicators of migrations. Specifically, we focus on novel conversation buzz and insecurity predictors based on Twitter data.


\subsubsection{Data Collection and Processing}\label{sec:data}

\begin{figure}
    \centering
    \includegraphics[scale=0.54]{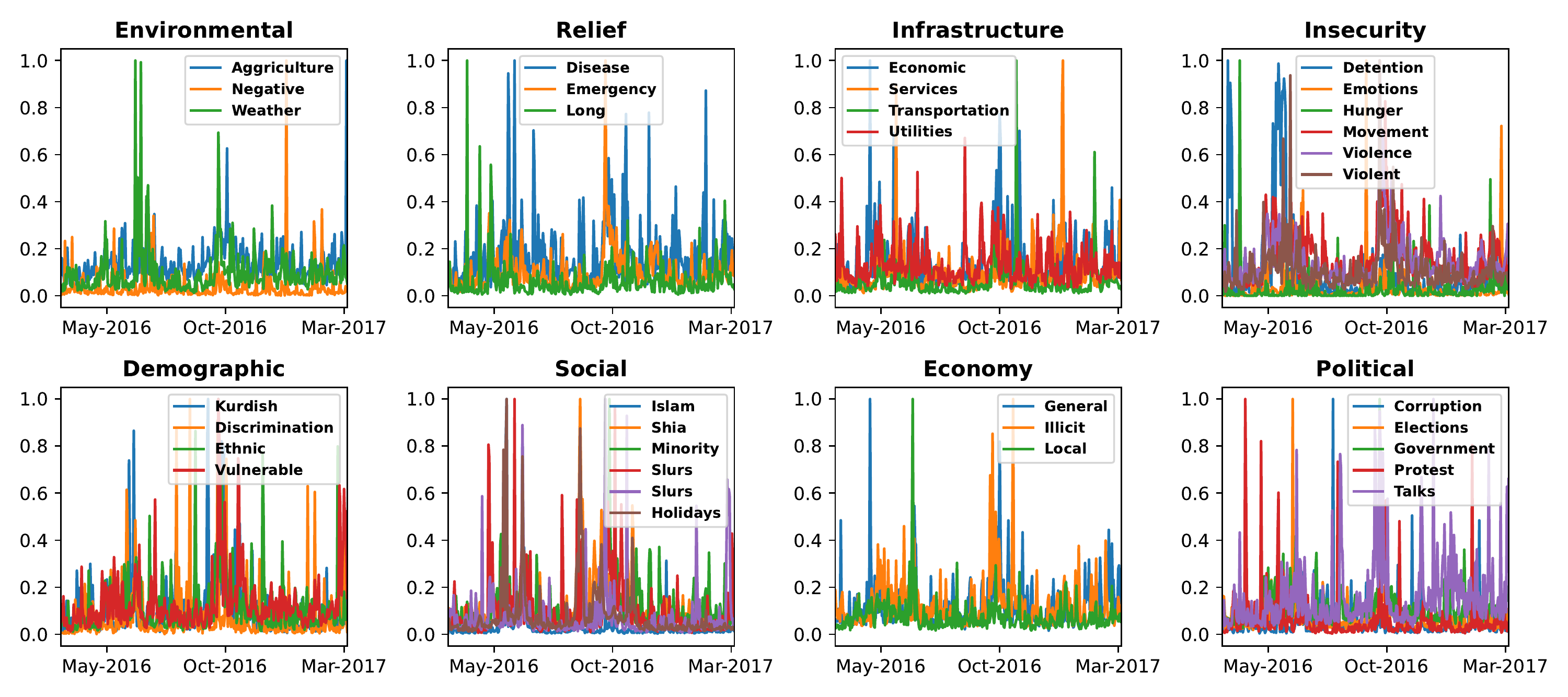}
    \includegraphics[scale=0.54]{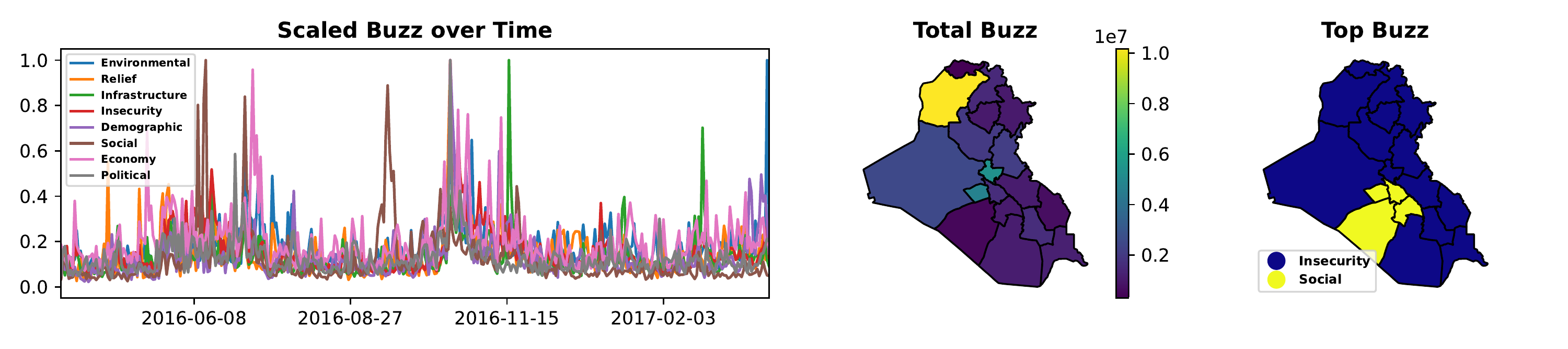}
    \includegraphics[scale=0.54]{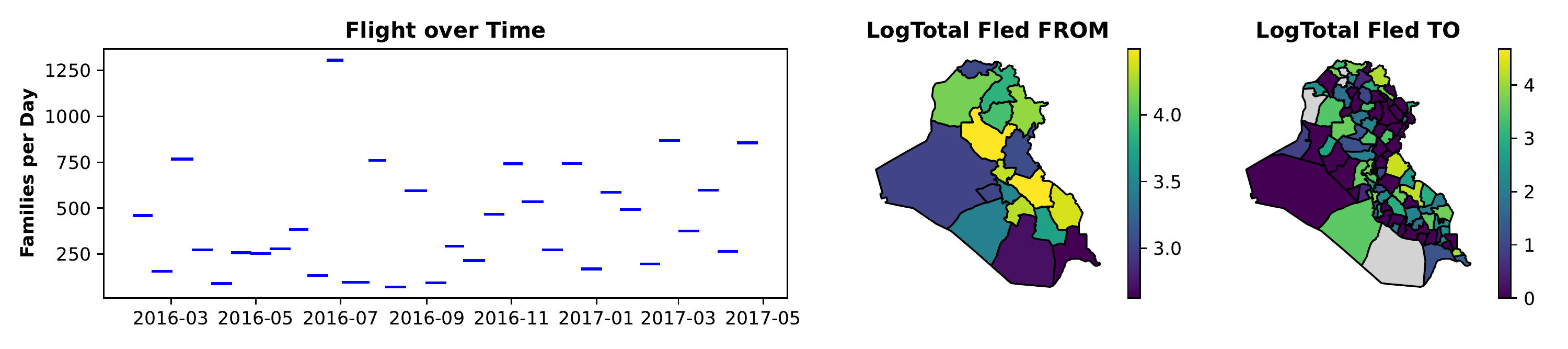}
    \caption{\textbf{Iraq Dataset:} \textit{Top Two Rows:} Buzz Variables broken up by category and normalized to have max value 1. \textit{Bottom Left Quadrant:} The cumulative buzz by each category is shown above the UN IOM data giving the families fleeing during each period. \textit{Bottom Right Quadrant:} Spatial totals for Buzz (upper row) and flow (lower row); gray districts not observed.}
    \label{fig:buzz_time}
\end{figure}

\noindent\textbf{Social Media Indicators} We collected Arabic language tweets between March 2016 and June 2017 from Twitter's Streaming Application Program Interface (API).
Buzz variables were formed by monitoring counts of keywords selected by Arabic language experts related to various migration-relevant topics, such as ethnic angst and discussion about the economy.
We also formed sentiment indicators.
Variables such as Social buzz for different demographic groups, (e.g. \texttt{Social-shia\_buzz}), Insecurity, (e.g. \texttt{Insecurity-emotions}), Environmental (e.g. \texttt{Environmental-weather}) are examples of the 68 variables identified.

\noindent\textbf{Movement Data} We use the flow data provided by the International Organization for Migration's Iraq Displacement Tracking Matrix, which gives stocks of displaced families in one of 97 Districts, as well as which of the 18 Governorates they originated from.
We differenced these data with respect to time to estimate flows.
The counts are aggregated over variable length periods, with the shortest period consisting of 11 days and longest of 31 days. There are 48 total time periods in our study, leading to a total of 57,024 observations.
96.4\% of response values are zero, with a mean of 6.71 families displaced for all time across origin-destination pairs.

\subsubsection{New Challenges for the Gravity Model}\label{sec:gravity}

Gravity models of migration \citep{Poprawe2015,Karemera2000,Ramos2017} measure the exchange between the origin $o$ and destination $d$ for a given time period $t$, denoted as $y_{t,o,d}$; in our case study $y_{t,o,d}$ is measured in number of families.
The gravity model specifies a generalized linear model relating $\mathbf{y}\in\mathbb{R}^{N_oN_dT}$ to predictors $\mathbf{X}\in\mathbb{R}^{N_oN_dT\times P}$  as well as \textit{resistance to exchange} terms $\omega_{t}$, $\omega_{o}$, $\omega_{d}$ which represent the effects associated with each time periods, origin and destination.
Poisson gravity models have been fit using MCMC for hospital patient flows \citep{Congdon2000} and trade \citep{Ranjan2007}.
\cite{Chen2018} perform explicit Bayesian Model Selection on a gravity model of trade via Normal approximations to their likelihood.
The fine timescale of our data introduces temporal misalignment between cause and effect: if an event which causes migration is recorded by our Twitter Buzz variables on a given day, we do not expect migrants to register at a camp the same day. 
At the yearly scale which gravity models are typically deployed, this is not a concern, but we need to develop a mechanism for estimating lags in this specific case study. 

\subsection{Spatiotemporal ZINB Gravity Model with Gaussian Lags}\label{sec:app_model}

In this section, we develop the following hierarchical Bayesian model to relate the organic variables to the observed migration:

\begin{align}
    &\lambda_p \overset{\mathrm{iid}}{\sim} \mathrm{C}(0,1)^+ \,\, \mathrm{for} \,\, p \in \{1, \ldots, P\}\\
    &\beta_p|\lambda_p \overset{\mathrm{indep}}{\sim} \mathrm{L}(0,\frac{1}{\tau\lambda_p})^+ \,\, \mathrm{for} \,\, p \in \{1, \ldots, P\}\\
    &\phi_\mathcal{I} \sim \mathrm{logitN}(a_\phi, b_\phi) \,\, \mathrm{for} \,\, \mathcal{I} \in \{\mathcal{T},\mathcal{O},\mathcal{D}\}\\
    &\rho^2_\mathcal{I} \sim \mathrm{logN}(a_\rho, b_\rho) \,\, \mathrm{for} \,\, \mathcal{I}  \in \{\mathcal{T},\mathcal{O},\mathcal{D}\}\\
    &\boldsymbol \omega_\mathcal{I}|\phi_\mathcal{I},\rho^2_\mathcal{I} \sim \mathrm{N}\big(\mathbf{0},\rho^2_\mathcal{I}\Sigma_\mathcal{I}(\phi_\mathcal{I})\big) \,\, \mathrm{for} \,\, \mathcal{I} \in \{\mathcal{T},\mathcal{O},\mathcal{D}\}\\
    &\pi \sim \mathrm{logitN}(a_\pi, b_\pi)\\
    & y_{t,o,d} \sim \begin{cases}
    \delta_0 & \mathrm{w.p.} \,\,\, \pi \\
    \mathrm{N.B.}(g(\mathbf{z}_{t,o,d}(\mu_l, \sigma_l)^\top\beta + \omega_t + \omega_o + \omega_d), \frac{l_{t}}{\sigma^2})\footnote{Here, the negative binomial is parameterized by probability of success $p\in[0,1]$ and number of failures $r$.} & \mathrm{o.w.}
    \end{cases}
\end{align}
where 
$\phi_\mathcal{T}\in[0,1]$ is the correlation between successive temporal effects $\boldsymbol\omega_t$ giving the AR1 correlation matrix $\Sigma_t(\phi_t)$, while $\phi_{o},\phi_{d}$ are the correlations for the origin and destination spatial effects $\boldsymbol\omega_\mathcal{O},\boldsymbol\omega_\mathcal{D}$, respectively, which lie between the inverse maximum and minimum eigenvalues of the adjacency matrix \citep{de2012bayesian}.
These matrices are scaled by variance terms $\rho_\mathcal{I}$ endowed with LogNormal priors.
The number of families moving between two regions on a given day is given by a negative binomial distribution parameterized by the canonical link transformation $g$ of our linear predictions.
The positive ``number of failures" parameter controlling overdispersion is estimated by $\frac{1}{\sigma^2}$ and is scaled by the duration of the time period $l_t$.
Due to the many of zeros, we include the possibility of zero inflation via a parameter $\pi$ with logitNormal prior.
We account for unknown and variable lag between cause and effect through a ``flattened" Gaussian filter with parameter $\mu_l$ determining how far back to shift the time series and the parameter $\sigma_l$ determining how far outside the bounds of the period length to integrate and producing the aggregated effects $\mathbf{z}$ (Supplementary
\if\arxiv 1
5).
\else
\ref{sec:app_imp}).
\fi
We use a neighborhood structure that has weight 1 if two regions touch and 0 otherwise.

\subsubsection{Variational Distribution and Inference}

We optimize the divergence via the VISTA algorithm.
We place logitNormal variational distributions on $\phi_t,\phi_o,\phi_d,\pi$ and logNormal variationals on $\rho_t,\rho_o,\rho_d,\sigma^2$. The KL divergence between prior and variational in all of these cases is given by the KL divergence of the underlying normal distributions and thus is available in closed form. 
The $\omega$'s are given a Normal variational distribution.
While the KL divergence cannot completely be determined in closed form due to the spatiotemporal correlation parameters $\boldsymbol\phi$, the $\omega$ and $\rho$ integrals may be computed analytically, and this expression may be evaluated for Monte Carlo samples of $\phi$ without any matrix decompositions beyond an initial decomposition of the neighborhood structure matrix (Supplementary
\if\arxiv 1
5).
\else
\ref{sec:app_imp}).
\fi
As is done in the context of Gibbs sampling for CAR models, we ``center on the fly" \citep{ferreira2019limiting} by projecting the variational mean of $\boldsymbol\omega_{\mathcal{I}}$ to have mean 0 after each iteration: $\boldsymbol\eta_{\omega_{\mathcal{I}}} \gets \boldsymbol\eta_{\omega_{\mathcal{I}}}- \mathrm{mean}(\boldsymbol\eta_{\omega_{\mathcal{I}}})\mathbf{1}$.

\subsection{Results} \label{sec:results}


\begin{figure}
    \centering
    \includegraphics[scale=0.7]{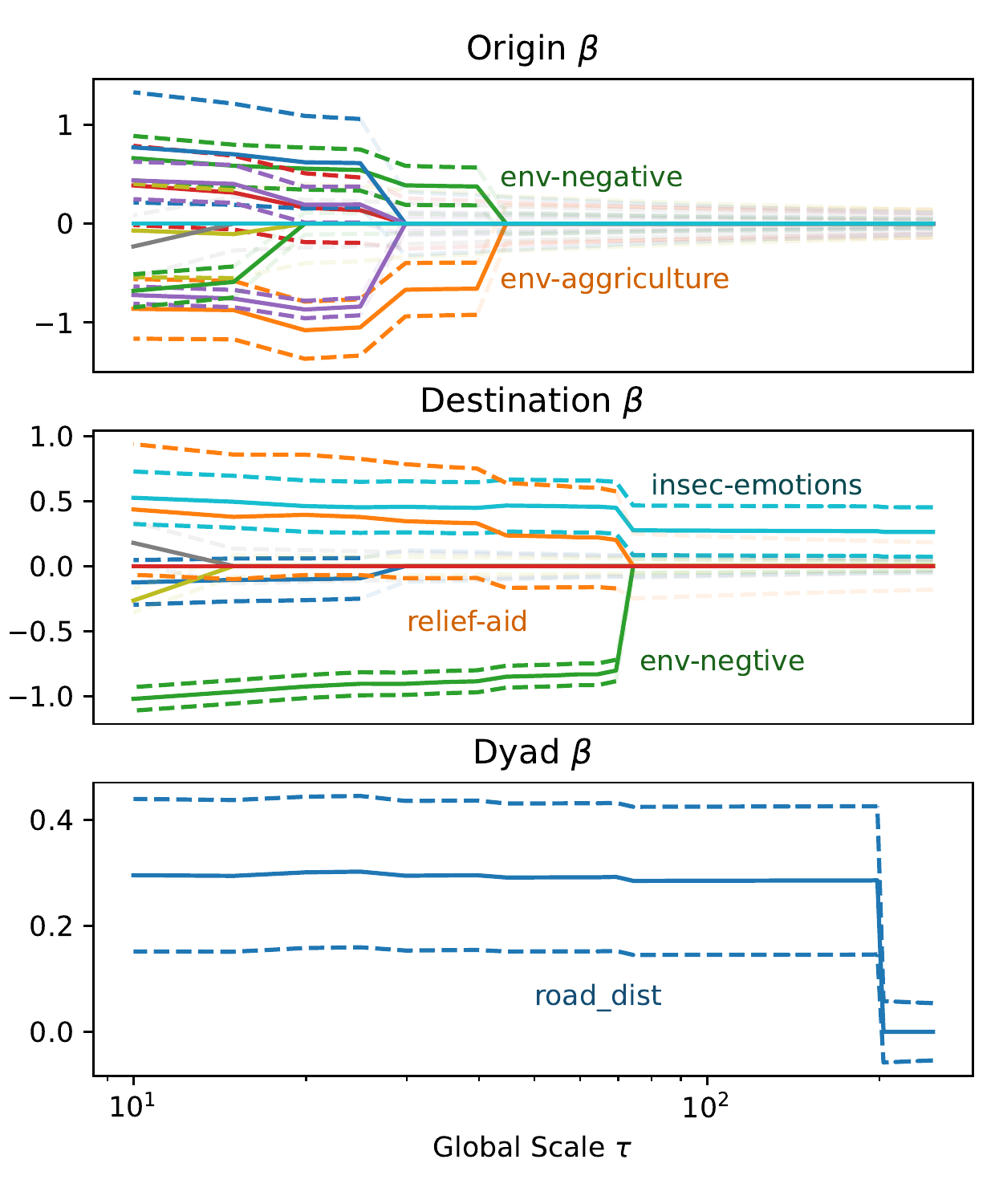}
    \includegraphics[scale=0.85]{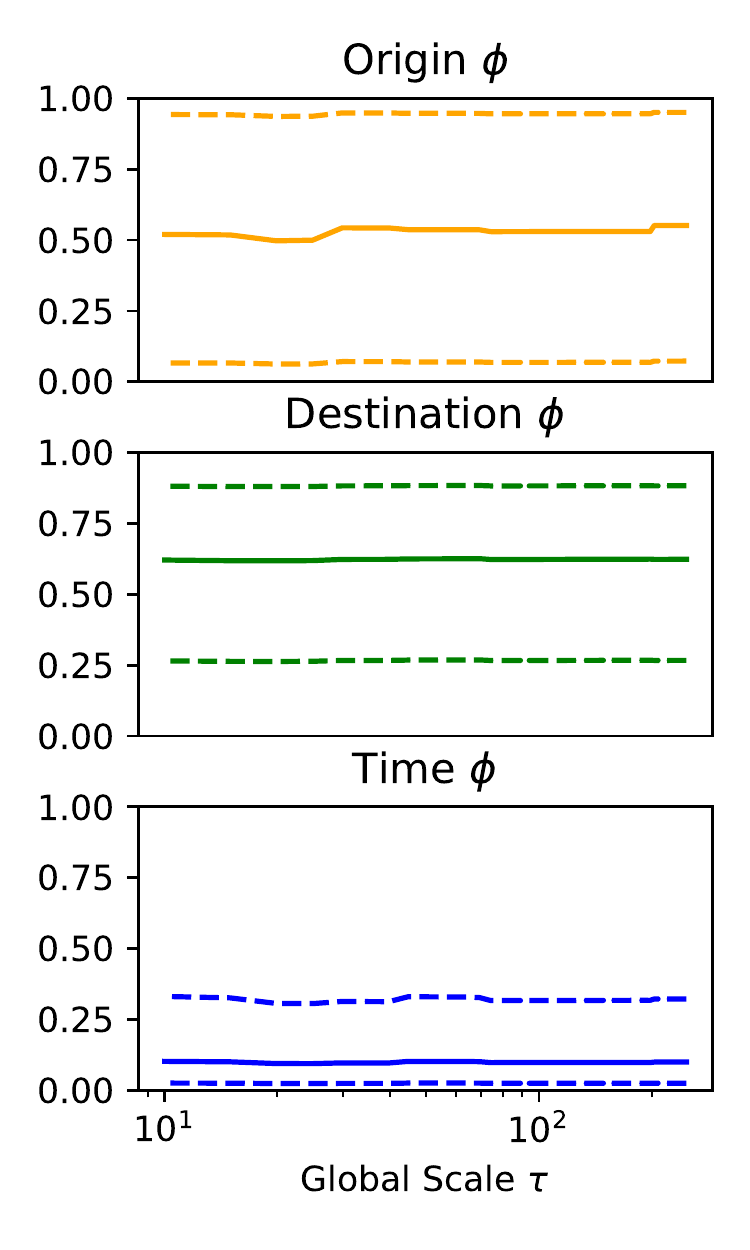}
    \includegraphics[scale=0.9]{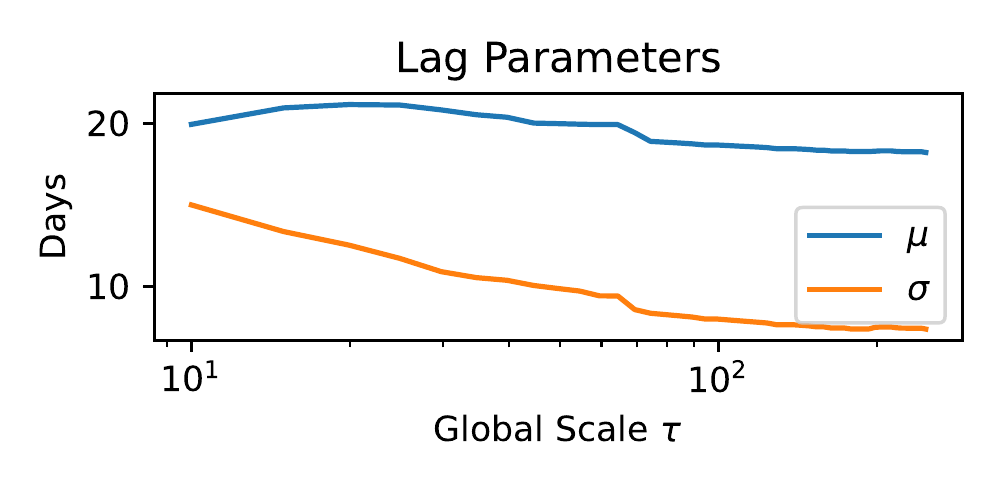}
    \includegraphics[scale=0.9]{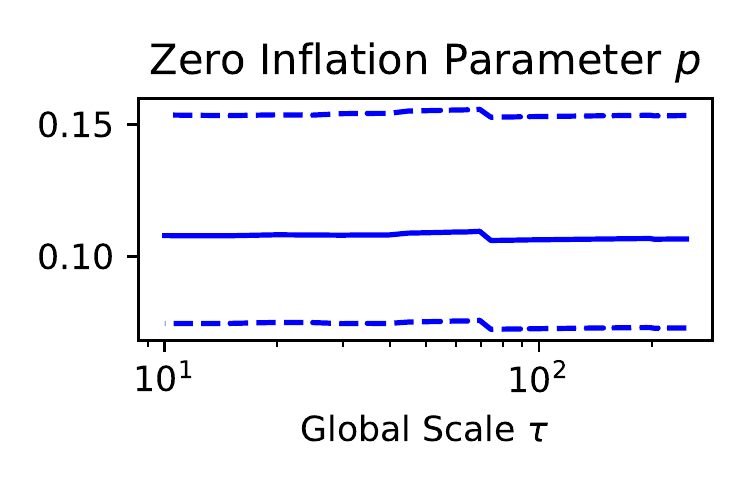}
    \caption{\textbf{Parameter Trajectories vs $\tau$.} \textit{Top Left:} Regression Coefficient Variational Trajectory; dotted lines give 95\% credible intervals; intervals are faded when the posterior mean is $0$. \textit{Top Right:} Spatiotemporal Correlations fairly constant.
    \textit{Bottom Left:} Gaussian convolution parameters vary with fixed effects.
    \textit{Bottom Left:} Zero inflation parameter is constant.
    }
    \label{fig:spt_traj}
\end{figure}

In this section, we apply our methodology to our case study. 
We begin by setting the regularization coefficient $\tau=250$, where the model chooses to include only the destination-level variable \texttt{insecurity-emotions-buz}. 
We then decreased this coefficient along a logarithmically spaced grid to $\tau=10$ with 100 points in between.
Figure \ref{fig:spt_traj} shows the trajectory borne out by the regression coefficients, random effect variances and lag and zero inflation parameters.
Begin with the top-left three panels, which show the variational distributions of the regression coefficients.
Qualitatively, we notice that variances of zeroed out coefficients decrease smoothly as the regularization strength increases, reflecting the increased prior distribution concentration.
The first variable to enter the model is destination \texttt{insecurity-emotions-buzz}, followed by \texttt{road-dist}, the dyadic distance between origin and destination, which actually seems to indicate that migration occurs over longer distances  (see Supplementary
\if\arxiv 1
5.4
\else
\ref{sec:app_antigravity} 
\fi
for more discussion).
Next \texttt{relief-aid} enters with a positive coefficient, indicating either that good relief aid attracts or is located where people are attracted, followed by \texttt{environment-negative} sentiment at the destination level and \texttt{environment-negative} and \texttt{environment-agriculture} at the origin level.

The lag-aggregation parameters are shown in the bottom left of Figure \ref{fig:spt_traj}. 
As the number of variables in the model increases, so too does the lookback window size $\sigma_l$, from about $9$ days to about $15$, while the number of days to look back travels from around $18$ to around $22$.
This varies from a model with sharper predictions from fewer variables to one with smoother predictions incorporating more variables and time.
The spatiotemporal correlation parameters' variational distributions are shown in the top 3 right panels of Figure \ref{fig:spt_traj}. 
The origin parameter is not well pinned down by our only 18 governorates, the Destination is also not strongly identified though it does rule out substantial negative spatial correlation, and the temporal correlation seems minimal.
Finally, the bottom right panel indicates that the zero-inflation parameter is also fairly constant with respect to $\tau$.


\begin{figure}
    \centering
    \includegraphics[scale=0.53]{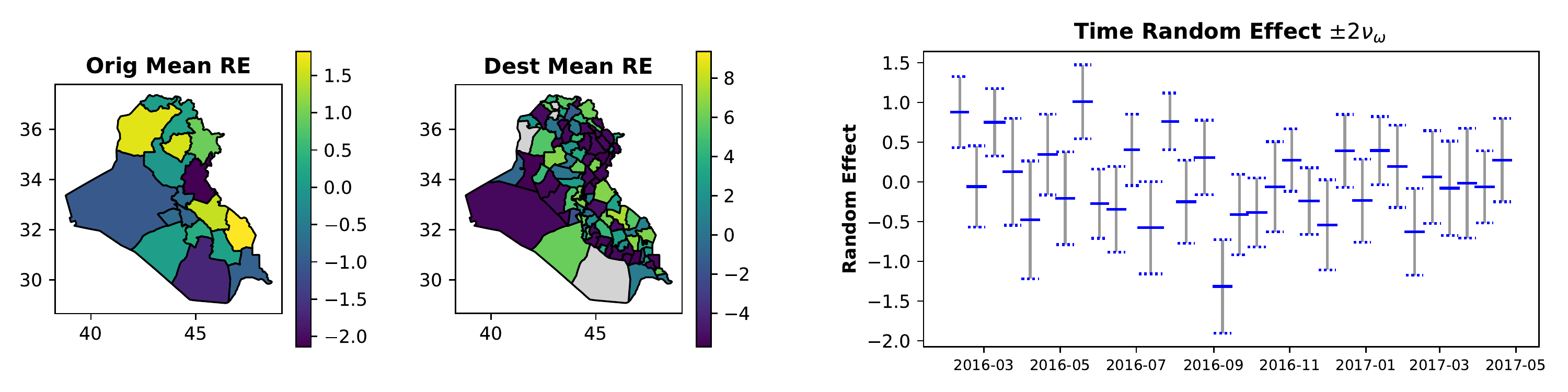}
    \setlength{\abovecaptionskip}{-30pt}
    \caption{\textbf{Spatiotemporal Random Effects:} At the largest $\tau$ value.}
    \label{fig:re_spt}
\end{figure}

The most positive origin random effects belongs to Ninewa Governorate with capital Mosul, which experienced urban combat against ISIL, followed by Maysan and Kirkuk governorates. 
For destinations, the Baghdad area has some of the highest effects, with its constituent Al-Sadr and Kut districts well separated from the rest.
On the lower end we see a different story: many districts which observed no inflow of migration have random effects which are experiencing separation and similar, negative values.
The time period with the largest random effect was from 2016-05-12 to 2016-05-26, a period which saw a series of ISIL bomb attacks which left over 100 dead in Baghdad (May 12) as well as combat between Iraqi and ISIL forces \citep{wilson2019Timeline}, while lowest random effect was 2016-09-01 to 2016-09-15, during which ISIL was more active abroad than in Iraq (ibid.).

\section{Discussion and Future Work}\label{sec:conclusion}

We began this article by examining the density of the Bayesian Lasso when $\lambda_p$ is given a hyperprior and treated as a variable to optimize.
This lead to a biconvex proximal operator that is well defined for sufficiently small step sizes.
We then showed how to plug this proximal operator into ISTA, producing the VISTA algorithm which allows for minimal bias even in the presence of high sparsity.
Next, we examined the asymptotic behavior of penalized likelihood estimates, finding that there exists a local penalized likelihood maximizer with the oracle property.
We then changed gears, showing how to deploy the penalty in the context of Variational Bayes and developing the Sparse Bayesian Lasso, which builds trajectories of entire variational distributions.
Simulation studies showed approximately nominal coverage of its credible intervals.
Finally, we deployed the VISTA algorithm and Sparse Bayesian Lasso on a sophisticated hierarchical model and were able to generate trajectories of variational distributions of all model parameters, varying the penalty strength.

We look forward to deploying the VISTA algorithm on other applications such as the inverse and imaging problems where ISTA made its name.
And we are even more excited to see what other possibilities viewing penalty coefficients as optimization quantities can open up.
In particular, Group Lasso \citep{yuan2006model} and Fused Lasso \citep{tibshirani2005sparsity}, among other Lasso analogs, enjoy fast inference via proximal gradient methods as well; perhaps the variable-penalty versions of these proximal operators are available in closed form.
Variable-coefficient low-rank penalties for matrices (e.g. \citet{koltchinskii2011nuclear}) and their associated proximal operators are also of future interest.
Recently, \citet{quaini2022proximal} developed basic inferential properties of estimators defined via a proximal operator; it would be interesting to see which of these properties apply to the biconvex proximal operator we have developed here.
We are also curious about different ways to motivate nonsmoothness in VB, perhaps via different divergences or discontinuous priors.

\if\blind0
\section*{Acknowledgements}

The authors gratefully acknowledge funding from the Massive Data Institute and McCourt Institute.
We also acknowledge the Georgetown Data Lab and MDI Technical Team for support creating migration indices, and Douglas Post for his preliminary analysis of the Iraq data.
We would like to thank Stephen Becker, Warren Hare, Nicholas Polson, Mahlet Tadesse, and Stefan Wild for valuable conversations and input.
Any errors are our own.
\fi



\bigskip
\begin{center}
{\large\bf SUPPLEMENTARY MATERIAL}
\end{center}

\begin{description}


\item[Python Package:] PyPI-package \texttt{sbl} containing code to perform the the methods outlined in this article.

\item[Supplementary Discussion and Results] PDF Document containing additional numerical results, mathematical details, and background.



\end{description}



\newcommand{\LL}{\mathcal{L}}
\newcommand{\bu}{\mathbf{u}}
\renewcommand\x{\mathbf{x}}
\newcommand\om{\boldsymbol\omega}
\renewcommand\theadalign{cl}
\renewcommand{\cellalign}{cl}
\renewcommand\theadfont{\bfseries}
\renewcommand\theadgape{\Gape[2pt]}
\renewcommand\cellgape{\Gape[1pt]}
\renewcommand\Po{\mathbf{P}_1}
\renewcommand{\y}{\mathbf{y}}
\renewcommand{\X}{\mathbf{X}}
\renewcommand{\bb}{\boldsymbol\beta}
\newcommand{\I}{\mathbf{I}}
\renewcommand\nvar{P}


\title{Supplementary Material: Sparse Bayesian Lasso via a Variable-Coefficient $\ell_1$ Penalty}

\maketitle

\newpage

\setcounter{theorem}{0}
\setcounter{lemma}{0}
\setcounter{section}{0}

\section{Extended Proofs for Section 2}\label{sec:app_proofs}

\begin{lemma}
    The marginal cost of P1 with respect to $\lambda$ (i.e. with $x$ profiled out) is the following piecewise quadratic expression:
\begin{equation}
     \underset{\lambda>0}{\mathrm{argmin}}  \begin{cases}
        \frac{1}{2}(\frac{1}{s_\lambda} - s_x)\lambda^2+(|x_0|-\frac{\lambda_0}{s_\lambda})\lambda + \frac{\lambda_0^2}{2s_\lambda} & \lambda < \frac{|x_0|}{s_x} \\
        \frac{(\lambda-\lambda_0)^2}{2s_\lambda}+\frac{x_0^2}{2s_x} & \lambda \geq \frac{|x_0|}{s_x}\,\, , \\
     \end{cases}
\end{equation}
where the changepoint $\lambda=\frac{|x_0|}{s_x}$ is the point where $\lambda$ is just large enough to push $x$ to zero.
\end{lemma}
\begin{proof}
    Convert to nested optimization and exploit the fact that the solution for known $\lambda$ is given by the soft thresholding operator:
\begin{align}
    &\underset{x\in\mathbb{R},\lambda>0}{\mathrm{argmin}} \,\,  \lambda |x|+\frac{(x-x_0)^2}{2s_x} + \frac{(\lambda-\lambda_0)^2}{2s_\lambda} \\
    &\iff 
    \underset{\lambda>0}{\mathrm{argmin}} \,\, \frac{(\lambda-\lambda_0)^2}{2s_\lambda} + \underset{x\in\mathbb{R}}{\mathrm{argmin}} \,\, \lambda |x|+\frac{(x-x_0)^2}{2s_x}  \\
    &\iff \underset{\lambda>0}{\mathrm{argmin}} \,\, \frac{(\lambda-\lambda_0)^2}{2s_\lambda} + \lambda (|x_0|-s_x\lambda)^+ +\frac{(|x_0|-s_x\lambda)^{+,2}-2|x_0|(|x_0|-s_x\lambda)^+}{2s_x}  \,\, . \label{eq:simpl}
\end{align}
\end{proof}

\begin{theorem}
    The optimizing $\lambda$ for the proximal program P1 is given by, when $s_xs_\lambda<1$:
    \begin{equation}\label{eq:prox1s}
        \lambda^* =\begin{cases} 
          \lambda_0 & \lambda_0 \geq \frac{|x_0|}{s_x} \\
          \frac{(\lambda_0-s_\lambda|x_0|)^+}{1-s_\lambda s_x} & o.w. \,\,\,\, ,
       \end{cases} 
    \end{equation}
    and by $\lambda^* = \mathbbm{1}_{\big[\frac{\lambda_0}{\sqrt{s_l}} > \frac{|x_0|}{\sqrt{s_x}}\big]} \lambda_0$ (here $\mathbbm{1}$ denotes the indicator function) otherwise.
In either case $x^* = (|x_0|-s_x\lambda^*)^+\mathrm{sgn}(x_0)$.
\end{theorem}
\begin{proof}
    We need only find the optimum of each interval. When $\lambda\geq \frac{|x_0|}{s_x}$, the optimum is simply as close as we can get to $\lambda_0$, namely $\lambda\gets\max[\lambda_0,\frac{x_0}{s_x}]$. 
On the other hand, when $\lambda \leq \frac{|x_0|}{s_x}$, if $s_xs_\lambda<1$, the optimum is as close as we can get to the stationary point $\frac{(\lambda_0-s_\lambda|x_0|)}{1-s_\lambda s_x}$, explicitly $\lambda\gets\min[\frac{(\lambda_0-s_\lambda|x_0|)^+}{1-s_\lambda s_x},\frac{|x_0|}{s_x}]$. When $s_xs_\lambda\geq1$, however, the solution is at one of the interval boundaries $[0,\frac{|x_0|}{s_x}]$; the boundaries have costs of $\frac{\lambda_0^2}{2s_\lambda}$ and $\frac{(\frac{|x_0|}{s_x}-\lambda_0)^2}{2s_\lambda}+\frac{x_0^2}{2s_x}$, respectively, and so we choose $\lambda\gets 0$ if $\frac{\lambda_0^2}{2s_\lambda}<\frac{(\frac{|x_0|}{s_x}-\lambda_0)^2}{2s_\lambda}+\frac{x_0^2}{2s_x}$ and $\lambda\gets\frac{|x_0|}{s_x}$ otherwise. But the cost at $\lambda=\lambda_0$ is only $\frac{x_0^2}{2s_x}$, so the choice is between $0$ and $\lambda_0$ with costs $\frac{\lambda_0^2}{2s_\lambda}$ and $\frac{x_0^2}{2s_x}$. 
\end{proof}

\begin{lemma}
	The following hold, where $\lambda^*$ denotes the optimizing $\lambda$, and is formally a function of $\tau$ and $|\beta|$:
 \begin{multicols}{2}
	\begin{enumerate}
		\item $\lambda^* = \frac{1}{\tau|\beta| + \rho'(\lambda^*)}$.
		\item $\frac{\partial \lambda^*}{\partial |\beta|} = -\frac{\tau}{\frac{1}{\lambda^{*2}} + \rho''(\lambda)}$.
		\item $g_{\tau}'(|\beta|) =  \tau \lambda^*$.
		\item $g_{\tau}''(|\beta|) =  -\frac{\tau_n^2}{(\tau_n+\rho'(\lambda^*))^2+\rho''(\lambda^*)}$.
	\end{enumerate}
 \end{multicols}
\end{lemma}
\begin{proof}
    It will be convenient to develop notation for the cost function inside our penalty: $g_\tau(|\beta|) = \underset{\lambda>0}{\min} \,\, \big[\tau\lambda|\beta| - \log\lambda + \rho(\lambda)\big] := \underset{\lambda>0}{\min} \,\, c^p(|\beta|,\lambda)$.
    For 1, since $\lambda^*$ is the optimizing $\lambda$, and due to the $-\log\lambda$ constraining the optimum to be an interior point, we know that $0=\frac{\partial}{\partial\lambda} \big[\tau\lambda|\beta| - \log\lambda + \rho(\lambda)\big] = \tau|\beta| - \frac{1}{\lambda} + \rho'(\lambda)$.
    For 2, we can use implicit differentiation on this same equation.
    For 3, we simply note that $\frac{\partial}{\partial|\beta|}g_{\tau}(|\beta|) = \frac{\partial}{\partial|\beta|} \big[\tau\lambda^*|\beta| - \log\lambda^* + \rho(\lambda^*)\big] = \tau\lambda^* + \frac{\partial\lambda^*}{\partial|\beta|}\frac{\partial c^p(|\beta|,\lambda)}{\partial\lambda}\Bigr|_{\lambda^*} = \tau\lambda^*$. 
    4 proceeds by differentiating 3 and plugging in 1 and 2.
\end{proof}

\begin{theorem}
    Assume that the logarithmic derivative of the hyperprior density on $\lambda$ is bounded ($|\rho'(\lambda)|\leq M_1 \,\, \forall \lambda\geq0$) and that the density is decreasing on $(0,\infty)$. Then:
    \begin{enumerate}
        \item $g'_{\tau}(|\beta|)\approx\frac{1}{|\beta|}$ for large $\beta$.
        \item The minimum of $|\beta|+g'_{\tau}(|\beta|)$ is achieved at $\beta=0$ with value $\lambda_a\tau$.
    \end{enumerate}
\end{theorem}
\begin{proof}
    For 1, $g'(|\beta|) = \frac{\tau}{\tau|\beta|+\rho(\lambda^*)}$, and since $\rho(\lambda)$ is bounded, $\underset{|\beta|\to\infty}{\lim}\frac{\tau}{\tau|\beta|+\rho(\lambda^*)} = \frac{1}{|\beta|}$.
    For 2, let $\lambda_a$ be the $\lambda$ such that $\frac{1}{\lambda_a} = \rho'(\lambda_a)$ (which is unique by the assumption that $\rho$ is increasing). Note that $\lambda_a \leq \lambda^*$ and $\lambda_a = \lambda^*(0)$. So each term of $|\beta|+\lambda^*\tau$ is decreasing in $|\beta|$ individually, and so the minimum of their sum must occur at $0$, yielding value $\lambda_a\tau$.
\end{proof}

\begin{theorem}
    Let $\tau_n=n\tau_0$ for $\tau_0>0$, and further assume that $|\rho''(|\lambda|)|<M_2$ (bounded second logarithmic derivative). Then, under the following standard regularity conditions on the likelihood:
    \begin{enumerate}
        \item The data $\mathbf{y}_i$ are i.i.d. with density function $f(\y;\bb)$ providing for common support and model identifiability. 
        We assume it has a score function with expectation zero $\mathbb{E}_{\bb}\Big[\nabla_{\bb} \log f(\y;\bb)\Big] = \mathbf{0}$ and a Fisher information expressible in terms of second derivatives: $I(\bb) = \mathbb{E}_{\bb}\Big[\nabla_{\bb}^2 \log f(\y;\bb) \Big]$. 
        \item The information matrix is finite and positive definite when $\bb=\bb_0$, with $\bb_0$ the true parameter vector. 
        \item For some open subset $\mathcal{B}$ containing $\bb_0$,  for almost all $\y$, the density is thrice differentiable $\forall\bb\in\mathcal{B}$ and that $\Bigr|\frac{\partial^3 \log f(\y;\bb)}{\partial\beta_i\partial\beta_j\partial\beta_k}\Bigr|\leq M_{i,j,k}(\y)$, also over $\mathcal{B}$, where the functions $M$ are such that $\mathbb{E}_{\bb_0}[M_{i,j,k}(\y)]<\infty$.
    \end{enumerate} 
   there is a local minimum of $Q(\bb) = -L(\bb) + \sum_{p=1}^P g_{\tau_n}(|\beta_p|)$ that satisfies the following:
    \begin{enumerate}
        \item $\hat{\bb}_2=\mathbf{0}$ with probability approaching 1 as $n\to\infty$.
        \item $\hat{\bb}_1$ is asymptotically normal with covariance given approximately by $\frac{1}{n}I(\bb_1)$, the Fisher information matrix considering only active variables.
    \end{enumerate}
\end{theorem}
\begin{proof}

We begin by establishing the existence of sparse local minima. Let $\mathcal{A}_n = \{\bb_1' : ||\bb_1'-\bb_{10}|| < \frac{C_1}{\sqrt{n}}\}$. We want to show that, asymptotically, for  $\bb_1\in\mathcal{A}_n$:

\begin{equation}
     \underset{\bb_2}{\min} \,\, Q\Big(  {\begin{bmatrix} \bb_1\\ \bb_2 \end{bmatrix}}\Big) =  Q\Big(  {\begin{bmatrix} \bb_1\\ \mathbf{0} \end{bmatrix}}\Big)
\end{equation}

We can do this by showing that $\textrm{sgn}(\frac{\partial Q(\bb)}{\partial \beta_j})=\textrm{sgn}(\beta_j)$ for $j>r$ and for $\bb_2$ sufficiently small.
Starting with the expression for the gradient and applying a series expansion on the likelihood, we see that:
\begin{align}
    \frac{\partial Q(\bb)}{\partial \beta_j} = -\frac{\partial L(\bb)}{\partial \beta_j} + g_{\tau_n}'(|\beta_j|)\textrm{sgn}(\beta_j) = -\frac{\partial L(\bb_0)}{\partial\beta_j} - \sum_{p=1}^P \frac{\partial^2 L(\bb_0)}{\partial\beta_j\partial\beta_p}(\beta_l-\beta_{l,0}) \\ 
    -\sum_{p_1,p_2=1}^P \frac{\partial^3 L(\bb^*)}{\partial\beta_j\partial\beta_{p_1}\partial\beta_{p_2}} (\beta_{p_1}-\beta_{p_1,0})(\beta_{p_2}-\beta_{p_2,0}) + g'_{\tau_n}(|\beta_j|)\textrm{sgn}(\beta_j)
\end{align}
where $\bb^*$ is associated with the Cauchy form for the remainder. 
Since $||\bb_1-\bb_{1,0}||=O_p(\frac{1}{\sqrt{n}})$ by assumption and plugging in the expression for $g_{\tau_n}'(|\beta_j|)$ from Lemma 2, when $\Vert\bb_2\Vert<\frac{1}{n}$, we have that:
\begin{align}
    \frac{\partial Q(\bb)}{\partial \beta_j} = O_p(\sqrt{n}) + \frac{n\tau_0}{n\tau_0|\beta_j|+\rho'(\lambda^*)}\textrm{sgn}(\beta_j) 
\end{align}
but when $\beta_j<\frac{1}{n}$, we have that $g_{\tau_n}'(|\beta_j|)\geq n\frac{\tau_0}{\tau_0+\rho'(\lambda^*)}$, which dominates the $O_p(\sqrt{n})$ likelihood term and means that $\textrm{sgn}(\beta_j)$ determines the sign of the partial derivative, and establishes the desired existence of sparse local minima.

Next, we want to show that within this set $\mathcal{A}_n$ which induces sparse minima in $\bb_2$ is a local minimizer of the likelihood with respect to $\bb_1$.
To this end, let vector $\bu\in\mathbb{R}^r$ be such that $\Vert\bu\Vert_2=C_1$.
We consider the difference in penalized likelihood between $\bb_{10}$ and $\bb_{10}+\frac{\bu}{\sqrt{N}}$, and again take the series expansion about $\bb_0$:
\begin{align}
    &Q(\bb_0+\frac{\bu}{\sqrt{n}}) - Q(\bb_0) = -L(\begin{bmatrix}\bb_{10}+\frac{\bu}{\sqrt{n}}\\\mathbf{0}\end{bmatrix}) + L(\begin{bmatrix}\bb_{10}\\\mathbf{0}\end{bmatrix}) + \sum_{i=1}^r g_{\tau_n}(|\beta_j|) - g_{\tau_n}(|\beta_{j0}|) \\
    &= -\frac{1}{\sqrt{n}}\frac{\partial L(\bb_0)}{\partial\bb_1}^\top\bu + \frac{1}{2} \bu^\top I_1(\bb_1)\bu\big[1+o_p(1)\big] + \sum_{p=1}^r \Big[\frac{1}{\sqrt{n}} g_{\tau_n}'(|\beta_p|)\textrm{sgn}(\beta_p) u_p + \frac{1}{n}g_{\tau_n}''(|\beta_p|) u_p^2\big[1+o_p(1)\big]\Big]
\end{align}
Notice that $g_{\tau_n}'(|\beta_p|) = \frac{n\tau_0}{n\tau_0|\beta_p|+\rho'(\lambda^*)}\overset{n\to\infty}{\to}\frac{1}{|\beta_p|}$ so $\frac{1}{\sqrt{n}}g_{\tau_n}'(|\beta_p|)\overset{n\to\infty}{\to}0$ and that $g_{\tau_n}''(|\beta_p|) = -\frac{\tau_n^2}{(\tau_n+\rho'(\lambda^*))^2+\rho''(\lambda^*)}$ also converges to a finite constant as $n\to\infty$, so $\frac{1}{n}g_{\tau_n}''(|\beta_p|)\to0$ as well.
For sufficiently large $C_1$, the quadratic likelihood term thus dominates all other terms and the difference is positive, and there is therefore a local minimizer with respect to $\bb_1$ within $\mathcal{A}_n$.

To determine its asymptotic distribution and conclude the proof, we examine the stationarity conditions for $j\in\{1,\ldots,r\}$ of this local optimum:

\begin{align}
    &0=\frac{\partial Q({\tiny\begin{bmatrix}\bb_1\\\mathbf{0}\end{bmatrix}})}{\partial\beta_j} \Bigr|_{\beta=\big(\substack{\hat{\bb_1}\\\mathbf{0}}\big)} = -\frac{\partial L(\bb)}{\partial \beta_j} \Bigr|_{\beta=\big(\substack{\hat{\bb_1}\\\mathbf{0}}\big)} + g_{\tau_n}'(|\hat\beta_j|)\textrm{sgn}(\beta_j) = \frac{\partial L(\bb_0)}{\partial \beta_j} \\
    &+ \sum_{p=1}^r \Big[\frac{\partial L(\bb_0)}{\partial\beta_j\partial\beta_p}+o_p(1)\Big](\hat\beta_p-\beta_{p0}) + g'_{\tau_n}(|\beta_j|)\textrm{sgn}(\beta_j) + \Big[g_{\tau_n}''(|\beta_j|)+o_p(1) \Big](\hat{\beta_j}-\beta_{j,0}) \,\, .
\end{align}
Slutsky's together with the Central Limit theorems thus give us that:
\begin{equation}
    \sqrt{n}\big(I_1(\bb_{1,0})+ \frac{1}{n} \mathbf{D}\big)\Big(\hat\bb_1-\bb_{1,0} + \frac{1}{n}\big(I_1(\bb_{1,0})+\frac{1}{n}\mathbf{D}\big)^{-1}\mathbf{g}\Big)\overset{d}{\to}\mathcal{N}\big(\mathbf{0},I_1(\bb_{1,0})\big) \,\, ,
\end{equation}
denoting:
\begin{equation}
    \mathbf{D} = \mathrm{diag}([ g''(|\beta_1|), \ldots, g''(|\beta_r|)]); \,\,\,\, \mathbf{g} = [g'_{\tau_n}(|\beta_j|)\textrm{sgn}(\beta_1), \ldots, g'_{\tau_n}(|\beta_j|)\textrm{sgn}(\beta_1)]
\end{equation}
\end{proof}

\section{Optimization Details}\label{sec:ap_optim}

\noindent\textbf{Step Size Selection with Trust Regions:}  We deploy an adaptive step size based on Trust Regions. 
The trust region framework involves comparing the observed function reduction with that implied, typically, by the second order model itself implied by our gradient descent:

\begin{equation}
	\hat{f}(\mathbf{x})\approx
	f(\mathbf{x}^t)+\nabla f (\mathbf{x}^t)^\top (\mathbf{x}-\mathbf{x}^t) +  ||\mathbf{x}^t-\mathbf{x}||_{\mathbf{C}}^2
\end{equation}
with the expected reduction in cost given by $\rho_{tr}:=\frac{f(\mathbf{x}^{t+1})-f(\mathbf{x}^t)}{\hat{f}(\mathbf{x}^{t+1})-\hat{f}(\mathbf{x}^t)}$, and where the norm $\mathbf{C}$ is given in our case by our diagonal step sizes: $||\mathbf{x}||_{\mathbf{C}}^2 = \sum_{p=1}^{P} \frac{x_p^2}{\eta_p}$.

In our case, our proximal approach means we have to include our regularization term in the trust region surrogate, and burden our notation by splitting the optimization variables into unpenalized variables $\theta$, penalized variables $\beta$, and penalizing variables $\lambda$, so $\mathbf{x}=[\theta,\beta,\lambda]$:
\begin{align*}
	\hat{f}([\theta, \beta, \lambda])\approx
	f([\theta^t, \beta^t, \lambda^t])+\nabla f ([\theta^t, \beta^t, \lambda^t])^\top ([\theta^t, \beta^t, \lambda^t]-[\theta, \beta, \lambda])   \\+||[\theta^t, \beta^t, \lambda^t]-[\theta, \beta, \lambda]||_{\mathbf{C}}^2 + \sum_{p=1}^P \lambda_p |\beta_p|,
\end{align*}
but is otherwise identical. 
When the $\rho_{rt}$ is too small, it means the expected reduction was much less than the actual reduction, and the trust region shrinks.
Conversely, a ratio closer to $1$ is indicative of an overly cautious stepsize. 
A trust region procedure proceeds by shrinking the step size by some factor $m_s$ whenever $\rho_{tr}<r_l$ is too small, expanding it by $m_e$ if $\rho_{tr}>r_u$ is too big, and leaving it unchanged when $\rho_{tr}$ is within a prespecified interval $[l_\rho,u_\rho]$. 
We use the default parameters recommended by the McCormick School's optimization Wiki \footnote{\url{https://optimization.mccormick.northwestern.edu/index.php/Trust-region_methods}}, which are $m_s=0.25,m_e=2,r_l=0.25,r_u=0.75$.
Steps which worsen the cost function are rejected and the step size is again shrunk by $m_s$.

\noindent\textbf{Nesterov Acceleration:} Along with careful step-size selection, we implement Nesterov acceleration \citep{nesterov1983method}, which the unfamiliar reader should think of as a careful kind of momentum. 
Nesterov acceleration has been shown to make a big difference in proximal gradient methods in the form of the ISTA algorithm, which is called Fast ISTA \citep{beck2009fast} when endowed with Nesterov acceleration.

One point of implementation we should mention is that Nesterov's search direction does not guarantee descent, unlike a gradient or proximal gradient direction.
This can be problematic since we are using a trust region to select step sizes, which expects our model to provide descent directions for sufficiently small step sizes. 
Various strategies exist to deal with this; we take the approach of simply resetting the Nesterov variables each time the trust region shrinks 3 times in a row (and resetting the step size to its previous value).

\noindent\textbf{Preconditioning:} Nesterov Acceleration gets us a faster asymptotic convergence rate. 
The other half of speedy optimization is preconditioning, which essentially involves improving the constants in the convergence rate.
By preconditioner, we simply mean a positive definite linear transformation of the gradient: $\tilde{\mathbf{g}}\gets\mathbf{C}\mathbf{g}$.
We studied the Natural Gradient Method \citep{amari1998natural}, which uses the Fisher information associated with the variational distribution as a preconditioner, as well as ``bespoke" preconditioners for various likelihood terms.
But ultimately, we found maximum generalizability in simply using an exponential moving average of gradient norms as our preconditioner, \textit{\'a la} Adam \citep{kingma2015adam}.
It is very likely that significant improvements in execution time on general problems can be attained by better preconditioners.

\noindent\textbf{Ablation Study:} We demonstrate the importance of Nesterov acceleration and proper preconditioning on our working problem. As Figure \ref{fig:optim_map} shows, the gradient search direction with as large a step size as we observed to be stable is much slower than that with a trust region-selected step size.
In turn, Nesterov acceleration is much faster than the trust region alone.
And finally, preconditioning improves optimization speed even more.

\begin{figure}
	\centering
	\includegraphics{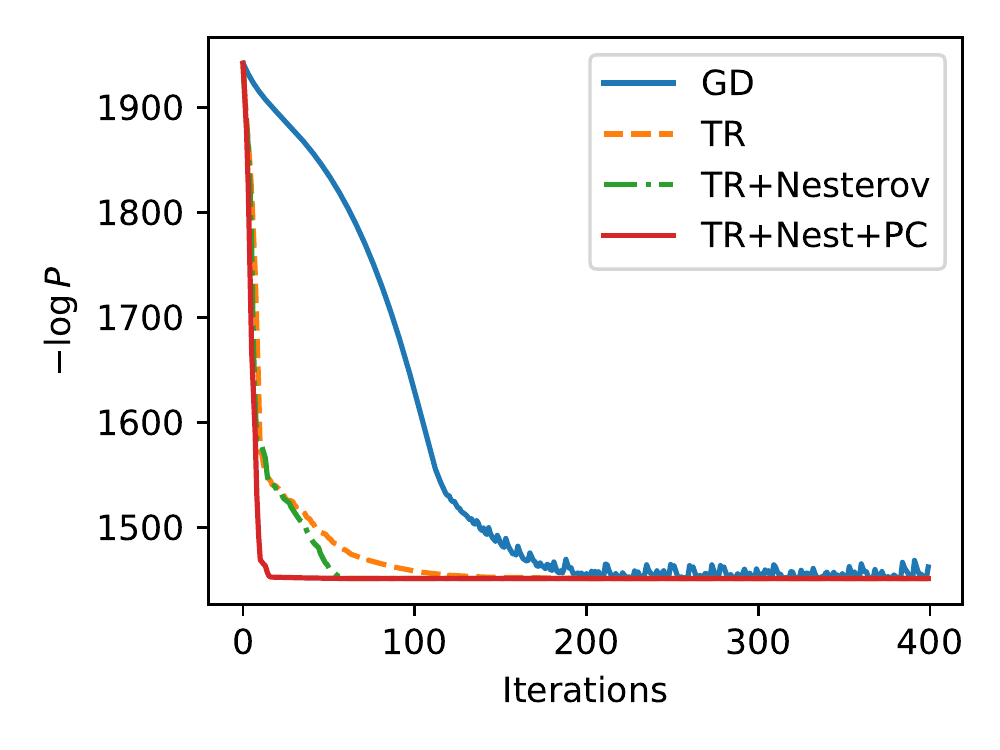}
	\caption{\textbf{Accelerating Gradient Descent:} Using adaptive step sizes via a Trust Region, better conditioning via a preconditioner, and acceleration via Nesterov's method yield much improved optimization performance on this problem.}
	\label{fig:optim_map}
\end{figure}

\section{SBL in  Linear Gaussian Models}\label{sec:app_lineargaus}





\begin{figure}
    \centering
	%
	\includegraphics[scale=0.78]{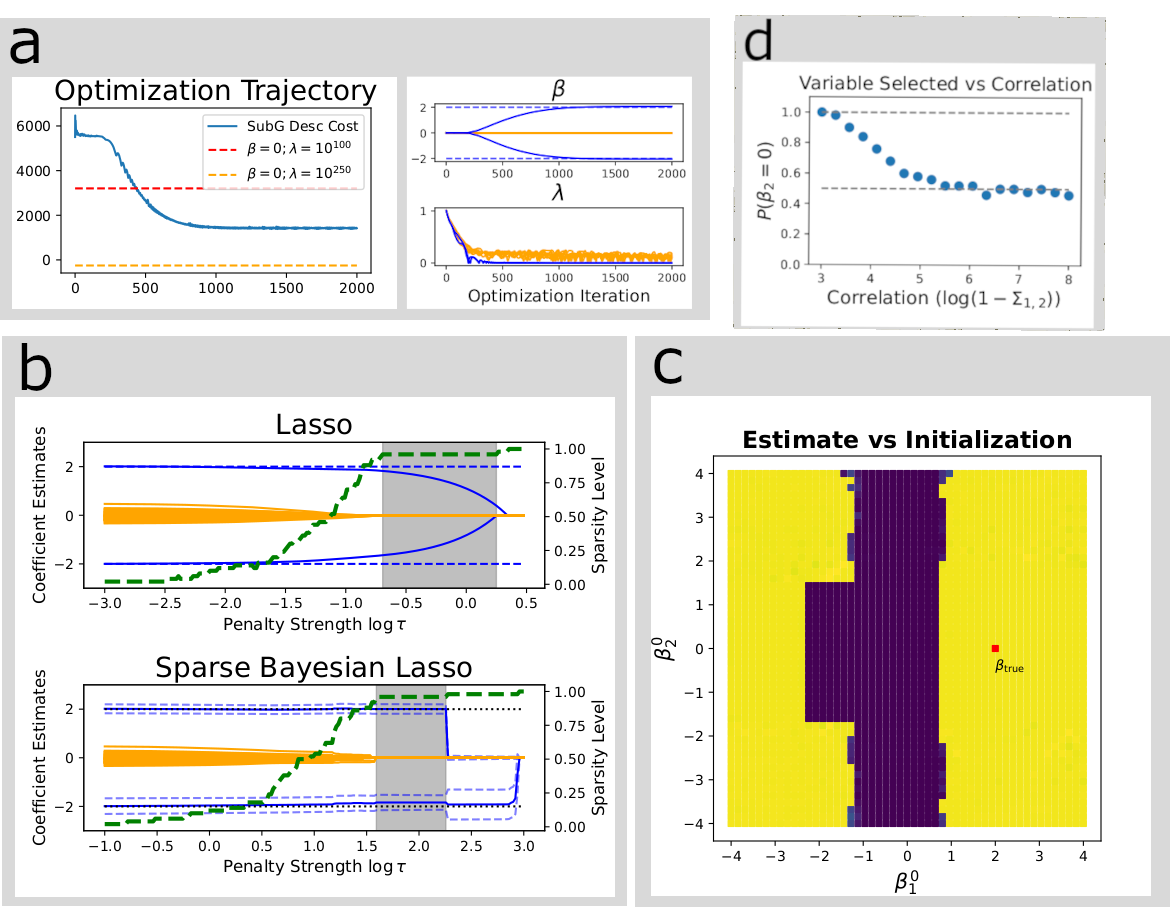}
    \caption{
    \textbf{Nonconvex Behavior.}
    \textit{Top Left and Mid:} MAP inference with uniform prior leads to reasonable local optima, but likelihood is unbounded.
    \textit{Lower and Mid Left:} Penalty $\tau$ varies: gray region is correct sparsity level, Lasso is biased within; SBL is debiased (like nonconvex penalties for likelihoods) and provides credible intervals.
    \textit{Top Right:} Instability in RNG for selection of correlated parameters.
    \textit{Bottom Right:} Cost vs Initialization, true parameter value in red.
    }
    \label{fig:map_unif}
\end{figure}
  
Denote by $\mathbf{P}_X$ the orthogonal projector onto the range of $\mathbf{X}$.
If we assume $\y\sim\mathrm{N}(\X\bb,\sigma^2 \I)$ with $\bb\sim Q_{\bb}$ and $\sigma^2\sim Q_{\sigma^2}$, then we have the following expected log likelihood: 
\begin{align}
    & -\underset{\bb\sim Q_{\bb},\sigma^2\sim Q_{\sigma^2}}{\mathbb{E}} [-\mathcal{L}(\bb,\sigma^2|\y)] =
    \underset{\bb\sim Q_{\bb},\sigma^2\sim Q_{\sigma^2}}{\mathbb{E}} [\frac{N}{2}\log\sigma^2 + \frac{||\y-\X\bb||_2^2}{2\sigma^2}] + C\\
    & = \underset{\sigma^2\sim Q_{\sigma^2}}{\mathbb{E}} [\mathrm{tr}[\frac{\X^\top\X}{2\sigma^2}\mathbb{V}_Q[\bb]]
    +(\mathbb{E}_Q[\bb]-\hat{\bb})^\top\frac{\X^\top\X}{2\sigma^2}(\mathbb{E}_Q[\bb]-\hat{\bb}) + \frac{\y^\top(\mathbf{I}-\mathbf{P}_X)\y}{2\sigma^2}
    +\frac{N}{2}\log\sigma^2] \\
    & = \underset{\sigma^2\sim Q_{\sigma^2}}{\mathbb{E}} [\frac{\mathrm{tr}[\X^\top\X\mathrm{diag}(2\nu^2)]
    +(\eta-\hat{\bb})^\top \X^\top\X (\eta-\hat{\bb}) + \y^\top(\mathbf{I}-\mathbf{P}_X)\y}{2\sigma^2}
    +\frac{N}{2}\log\sigma^2] \, .
\end{align}
To evaluate the expectation wrt $\sigma^2$, we'll define the numerator of the first term as $\xi$ which yields the following:
\begin{equation}
    = \underset{\sigma^2\sim Q_{\sigma^2}}{\mathbb{E}} [\frac{\xi}{2\sigma^2} +\frac{N}{2}\log\sigma^2] 
    = \frac{\xi}{2}\underset{\sigma^2\sim Q_{\sigma^2}}{\mathbb{E}}[\frac{1}{\sigma^2}] +\frac{N}{2}\underset{\sigma^2\sim Q_{\sigma^2}}{\mathbb{E}}[\log\sigma^2] \, .
\end{equation}
We see that if we want a closed form estimate of KL we need to choose $Q_{\bb}$ to have closed form mean and variance and $Q_{\sigma^2}$ to have a closed form reciprocal and logarithmic expectation.
The Gamma and Lognormal distribution satisfy this.
However, we observed difficulties differentiating through the Gamma simulation in \texttt{tensorflow\_probability}, so we settled on the Lognormal which had better numerical behavior.

Normal and lognormal variational-prior models have the same well known KL divergence expression.

\section{$\lambda$-Prior Specification in Orthogonal Designs} \label{sec:app_map}

\noindent\textbf{Working example} $\mathbf{X}\in\mathbb{R}^{N\times P}=\mathbb{R}^{250\times 10}$ with entries sampled from $i.i.d$ standard Gaussians, $\boldsymbol\beta\in\mathbb{R}^{50}$ with only 6 nonzero entries, of $[-2.5,-2,-1.5,1.5,2,2.5]$ (similar $\beta$ as \citet{Rockova2018}), and $\mathbf{y}=\X\boldsymbol\beta+\boldsymbol\epsilon$ with again iid standard normal $\epsilon_n$.
Figure \ref{fig:map_unif}, a shows a more readable problem with only $[-2,2]$ nonzero in $\beta$ and $N=100$.

We now consider the orthogonal linear case where the problem decomposes axis-by-axis, starting again by converting to a nested optimization marginalizing over $x$, yielding:
\begin{align}
    \underset{\beta\in\mathbb{R},\lambda\in\mathbb{R}^+}\min \frac{\sum_{i=1}^N(y_i-x_i\beta)^2}{2\sigma^2} +\tau\lambda|\beta| - \log\lambda - \log p_\lambda(\lambda) \\
    \iff \underset{\beta\in\mathbb{R},\lambda\in\mathbb{R}^+}\min \frac{(\beta-\hat{\beta})^2}{2\frac{\sigma^2}{\sum_{i=1}^N x_i^2}} +\tau\lambda|\beta|  - \tau\log\lambda - \log p_\lambda(\lambda) \,\, , 
\end{align}
where $\hat{\beta}$ gives the least squares estimate.
Again, the STO gives us profile cost for $\lambda$ with $\beta^*(\lambda) = (|\hat{\beta}|-\frac{\sigma^2\tau\lambda}{\sum_{i=1}^N x_i^2})^+\mathrm{sgn}(\hat{\beta})$. Using the fact that $(\beta^*(\lambda)-\hat{\beta})^2=\min(\frac{\sigma^2\tau\lambda}{\sum_{i=1}^N x_i^2},|\hat{\beta}|)^2$ 
, we can express the profile cost as a piecewise log-quadratic (plus the prior term):
\begin{align}
    &\underset{\lambda\in\mathbb{R}^+}{\min} \,\, \frac{\tau^2\min(\lambda, \frac{|\hat{\beta}|\sum_{i=1}^N x_i^2}{\tau\sigma^2})^2}{2} 
    + \tau\lambda \Big(|\hat{\beta}|-\frac{\sigma^2\tau\lambda}{\sum_{i=1}^Nx_i^2}\Big)^+ 
    - \log\lambda - \log p_\lambda(\lambda) \\
    &\iff \underset{\lambda\in\mathbb{R}^+}{\min}\begin{cases}
        \tau^2\bigg(\frac{1}{2} - \frac{\sigma^2}{\sum_{i=1}^N x_i^2}\bigg)\lambda^2 + \tau |\hat{\beta}|\lambda 
        -\log(\lambda)-
        \log p_\lambda(\lambda); & \lambda \leq \frac{|\hat{\beta}|\sum_{i=1}^N x_i^2}{\tau\sigma^2} \\
        \frac{\hat{\beta}^2[\sum_{i=1}^N x_i^2]^2}{2\sigma^4} - \log\lambda - \log p_\lambda(\lambda)
    \end{cases}
\end{align} 

\noindent\textbf{Uniform} $p_\lambda\propto 1$: 
Without a great idea of what prior we ought to put on $\lambda$, we might be tempted to place a uniform prior on it. 
However, this leads to an unbounded density as  $(\boldsymbol\beta,\boldsymbol\lambda)=(\mathbf{0},\infty)$: the term $\frac{||\y-\X\boldsymbol\beta||_2^2}{2\sigma^2}$ becomes $\frac{||\y||_2^2}{2\sigma^2}$, while the $-\log\lambda_p$ terms are unbounded.



Undeterred, we performed MAP inference with a uniform prior on our working problem using the Adam optimizer (\cite{kingma2015adam}) initialized with $\boldsymbol\beta=\mathbf{0}$, $\boldsymbol\lambda=\mathbf{1}$ and subgradient descent, where it is seemingly able to perform debiased variable selection (see Figure \ref{fig:map_unif}a). 
The ``zero $\boldsymbol\beta$ big $\boldsymbol\lambda$" strategy only surpasses the posterior density of local estimates on this particular example at $\lambda\approx10^{200}$; since the $\log\lambda$ term is unbounded, we expect this behavior for any $\mathbf{X},\mathbf{y},\sigma^2$. 
Evidently, local optimizations may converge to a reasonable local minimum despite this.
Though the biconvex penalty proposed here has its advantages over convex penalties, nonconvexity can be a treacherous property.

\noindent\textbf{Power Inverse} $p_\lambda\propto \frac{1}{\lambda^{a}}$: In this case, setting $\boldsymbol\beta=0$ leaves us with a cost of $\sum_{p=1}^P (a-1)\log\lambda_p$.
We see that for $a<1$, we get similar behavior to the uniform case, where cost is unbounded below as $\lambda\to0$. 
$a>1$ leads to a positive coefficient for the $\log\lambda$ term, in turn leading to the opposite behavior, where cost is unbounded below as $\lambda\to0$.
When $a=1$, the log terms cancel, and we are left with a cost of $\frac{||\y-\X\boldsymbol\beta||_2^2}{2\sigma^2}+\sum_{p=1}^P\lambda_p|\beta_p|$.
Without any motivation to stay positive, the $\lambda_p$ term will vanish, yielding the unpenalized problem.

    

\noindent\textbf{Half-Gaussian} $p_\lambda\propto e^{-\frac{(\lambda-m_{\lambda})^2}{b_\lambda^2}}$: Plugging in yields:

\begin{equation}
     \underset{\lambda\in\mathbb{R}^+}{\min}\begin{cases}
        \bigg(\tau^2(\frac{1}{2} - \frac{\sigma^2}{\sum_{i=1}^N x_i^2}) + \frac{1}{b_\lambda^2}\bigg)\lambda^2 + (\tau |\hat{\beta}|-2\frac{m_\lambda}{b_\lambda^2})\lambda 
        -\log(\lambda) +\frac{m_\lambda^2}{b_\lambda^2};
        & \lambda \leq \frac{|\hat{\beta}|\sum_{i=1}^N x_i^2}{\tau\sigma^2} \\
        \frac{\hat{\beta}^2[\sum_{i=1}^N x_i^2]^2}{2\sigma^4} - \log\lambda + \frac{(\lambda-m_\lambda)^2}{b_\lambda^2} & o.w.
    \end{cases}
\end{equation}
This function, like the proximal cost, may be nonconvex when 
$\frac{\sigma^2}{\sum_{i=1}^N x_i^2}>\frac{1}{\tau^2b_\lambda^2}+\frac{1}{2}$, 
as this is when the quadratic term is concave.
But the $-\log\lambda$ term contributes to convexity near the origin such that the overall function is convex until 
$\frac{1}{\lambda^2} > \tau\big(\frac{2\tau\sigma^2}{\sum_{i=1}^N x_i^2}-\tau\big)-\frac{1}{b^2_\lambda}$.
The stationary point of the upper convex region is 
$\lambda^*=\frac{m_\lambda+\sqrt{m_\lambda^2+2b_\lambda^2}}{2}$
(if this lies in the region), and is
\begin{equation}
    \frac{2 b_\lambda^2 s_x}{\beta  b_\lambda^2 \tau  s_x+\sqrt{s_x \left(b_\lambda^4 \tau ^2 \left(\left(\beta ^2+4\right)
   s_x-8 \sigma \right)-4 b_\lambda^2 s_x (\beta  m_\lambda \tau -2)+4 m_\lambda^2 s_x\right)}-2 m_\lambda s_x}
\end{equation}
in the lower region. 
The $\log\lambda$ term prevents optima at $\lambda=0$, so we need only check the cost at these two points to determine which is optimal.


\noindent\textbf{Half-Cauchy} $p_\lambda \propto \frac{1}{1+\frac{\lambda^2}{a}}$: We next turn our attention to the case $p_\lambda=\mathrm{C}(0,a_\lambda)^+$, a Half-Cauchy prior, in the style of the Horseshoe Prior (but with a Laplace rather than Normal conditional prior for $\beta$).
This leads us to the following cost function:
\begin{align}
    \underset{\lambda\geq 0}{\min} \begin{cases}
        \tau^2\bigg(\frac{1}{2} - \frac{\sigma^2}{\sum_{i=1}^N x_i^2}\bigg)\lambda^2
        + \tau\lambda|\hat{\beta}|+\log\Big[\frac{1+\frac{\lambda^2}{a_\lambda}}{\lambda}\Big]; & \lambda \in \lambda \leq \frac{|\hat{\beta}|\sum_{i=1}^N x_i^2}{\tau\sigma^2} \\
         \frac{\hat{\beta}^2}{2\sigma^4} + \log\Big[\frac{1+\frac{\lambda^2}{a_\lambda}}{\lambda}\Big] ; & o.w.
    \end{cases}
\end{align}

Finding the critical points of the function in $[0,\frac{|\hat{\beta}|\sum_{i=1}^N x_i^2}{\tau\sigma^2}]$ may be cast as solving a quartic polynomial equation.
While we are fortunate that such an equation can be solved in closed form, we have not been able to extract much intuition from it.
Nevertheless, the variable penalty $\ell_1$ proximal operator enables iterative numerical solution of this problem.

\section{Additional Experimental Results and Details}

\subsection{Simulation Study}

$\tau$ values were selected manually which gave good performance in simulation studies.
In practice, we recommend evaluating the SBL variational distribution for a whole spectrum of $\tau$ rather than choosing a particular one.

\subsubsection{Nonconvexity and Instability}

Since our penalty is nonconvex, the final parameters of an optimization are dependent on their initializations.
Figure \ref{fig:map_unif}c shows the final costs associated with a grid of initializations on a 2D toy problem, where $\beta=[2,0]$. The initial values are used on a grid of 51 points between [-4,4] and with $\tau$ set to 100.
We perform MAP inference to focus on the nonconvexity in the coefficients (rather than on  variational standard deviations).
We see that in this case, the model lands on the true solution when $\beta_1$ is initialized at large values.
For this simulation and throughout the article, $\lambda_p$ were all initialized at 1 and $\beta$ at 0.
The classical Lasso is famously unstable insofar as it essentially randomly selects from groups of correlated parameters.
We conduct a numerical experiment to confirm the inheritance of this property, using the same toy problem, but now with $\mathbf{X}$ data simulated from a correlated normal distribution with high correlation, and this time using the SBL.
Figure \ref{fig:map_unif}d shows how, starting at a correlation of $0.999$, as we approach a correlation of $1$, the model begins changing from nearly always selecting $\beta_1$ as nonzero to doing so only $50\%$ of the time.

\subsection{Additional Case Study Implementation Details}\label{sec:app_imp}

\begin{table}[h]
\small
    \centering
    \begin{tabular}{|c|c|c|c|c|}
        \hline
         Symbol & Generating Dist & Variational Dist & Name\\
         \hline
         $\beta$ & $\mathrm{L}(0,\frac{1}{\lambda_{\ell_1}})$ & $\mathrm{L}(\eta_\beta,\nu_\beta)$ & Regression Coefficients\\
         $\sigma^2$ & $\log\mathrm{N}(0,b_{\sigma})$ & $\log\mathrm{N}(\eta_\sigma,\nu_\sigma)$ & Overdispersion Parameter\\
         $\omega$ & $\mathrm{N}(0,\rho^2)$ & $\mathrm{N}(\eta_\omega,\nu_\omega)$ & Spatiotemporal Random Effects\\
         $\rho^2$ & $\log\mathrm{N}(0,b_{\rho})$ & $\log\mathrm{N}(\eta_\rho,\nu_\rho)$ & Random Effects Variance\\
         $\phi$ & $\mathrm{U}(0,1)$ & $\mathrm{logit}\mathrm{N}(\eta_\phi,\nu_\phi)$ & Spatial/Temporal Correlation\\
         $\pi$ & $\mathrm{U}(0,1)$ & $\mathrm{logit}\mathrm{N}(\eta_\pi,\nu_\pi)$ & Zero Inflation Probability\\
         \hline
    \end{tabular}
    \caption{Model Estimands along with their assumed prior and variational distributions.}
    \label{tab:estimands}
\end{table}

\subsubsection{Lagging and Aggregation}

In order to account for lag between tweet measurement and recording of a migration event, we use a Gaussian filter. 
Recall that we have periods $r$ of variable length, with beginning $l_r$ and end $u_r$, and duration $s_r:=u_r-l_r$. 
The filter, parameterized by a mean lookback time $\mu_l$ and smoothing window size $\sigma_l$, is defined as:
\begin{equation}
    z_{r,o,d,p} \propto \sum_{t'=t-w}^{t+w}\frac{f((|\frac{u_t+l_t}{2}-t'|-\frac{s_t}{2})^+,\mu_l,\sigma_l^2) x_{t',o,d,p}}{\sum_{t'=t-w}^{t+w} f((|\frac{u_t+l_t}{2}-t'|-\frac{s_t}{2})^+),\mu_l,\sigma_l^2)}
\end{equation}
where $w$ is the maximum window size, $f(.,\mu,\sigma^2)$ gives the Gaussian density function.
In other words, all tweets within the period $r$ when shifted by $\mu_l$ is assigned maximum density. 
Events outside of this window less density, with speed of decay give by $\sigma_l^2$

\subsubsection{Spatiotemporal KL Penalty Estimation}

In this section we describe how to form the KL terms related to the spatial random effects $\omega$, the spatial correlation parameter $\phi$ and scaling parameter $\rho^2$. 
We assume the CAR model $\omega \overset{P}{\sim} N(\mathbf{0},\rho(\I-\phi\mathbf{M})^{-1})$, where $\mathbf{M}$ is a spatial connectivity matrix.

We begin with the normal KL, denoting by $B$ the dimension of $\om$:
\begin{align}
    & \frac{1}{2} (\log|\Sigma_P|-\log|\Sigma_Q|) - B + B \mathbb{E}[\log\rho^2] + \big [ \mathrm{tr}[\Sigma_Q] + \mathbb{E}[\om]^\top\Sigma_P^{-1}\mathbb{E}[\om] \big] \mathbb{E}[\frac{1}{\sigma^2}] \,\, ,
\end{align}
and plugging in the following:
\begin{align}
    & \Sigma_P^{-1} = \I-\phi\mathbf{M} \\
    & \log|\Sigma_P| = -\sum_{i}\log(1-\phi\lambda_i(\mathbf{M})) \\
    & \log|\Sigma_Q| = -\sum_{i} \nu_{\omega_i} \, .
\end{align}
In the lognormal variational case, the two expectations of $\rho^2$ are available in closed form as $\mathbb{E}[\log\rho^2]=\eta_\rho$, $\mathbb{E}[\frac{1}{\rho^2}]=e^{-\eta_\rho+\frac{\nu_\rho}{2}}$.
It remains to integrate out $\phi$, which is not analytical.
This we accomplish via SAA, just as we estimate the negative log likelihood.

\subsection{Case Study Random Effects}

The variational estimates of the random effects are given in Figure \ref{fig:app_spt_traj}. 
The right set of panels, which give us the random effect variance and thus average magnitude, indicate that there is not much change in the magnitude of effects overall as $\tau$ is varied.
On the other hand, some individual random effects do undergo nontrivial changes as $\tau$ varies, and the four most changing random effects are highlighted in each figure.

\begin{figure}
    \centering
    \includegraphics[scale=0.6]{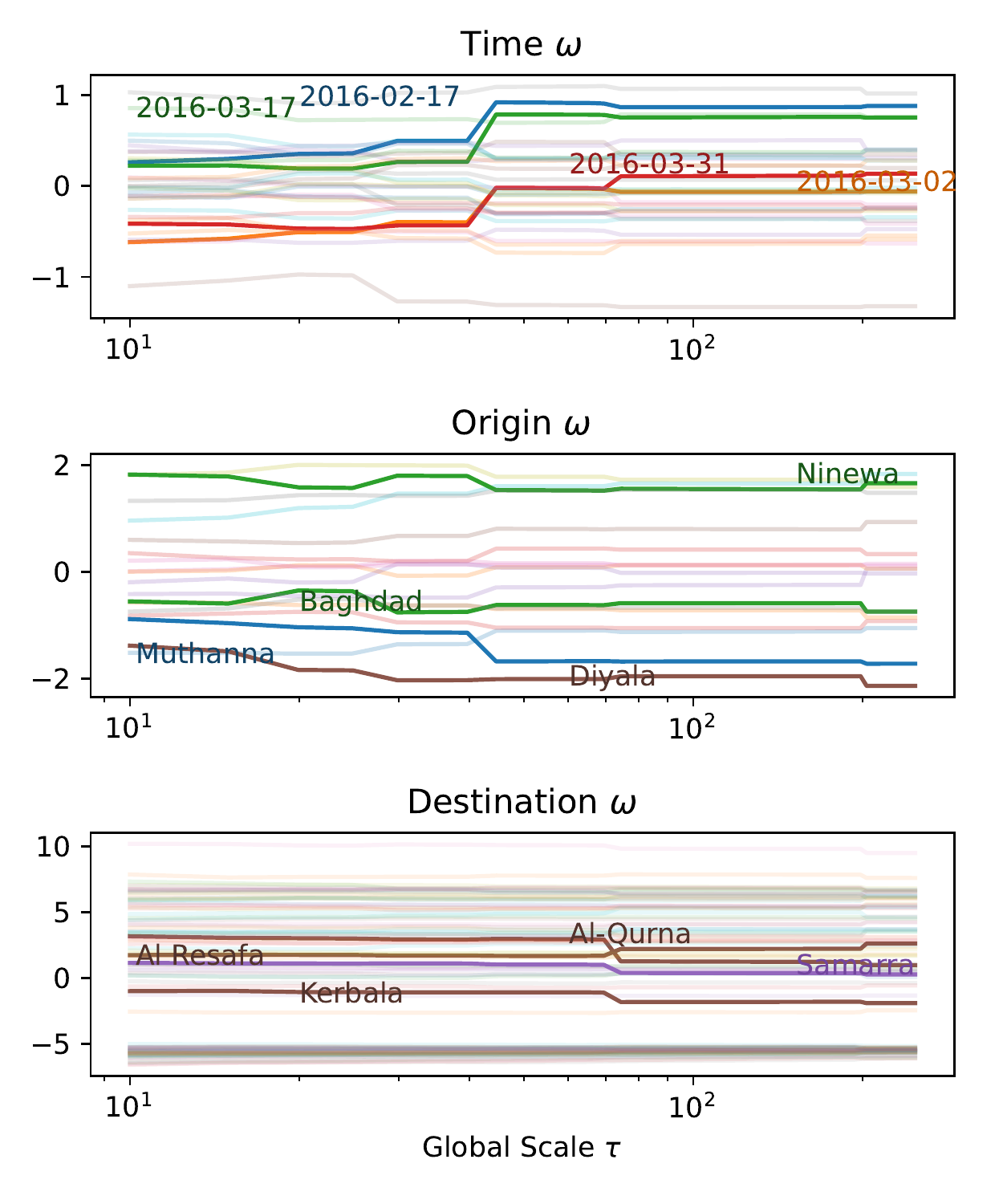}
    \includegraphics[scale=0.735]{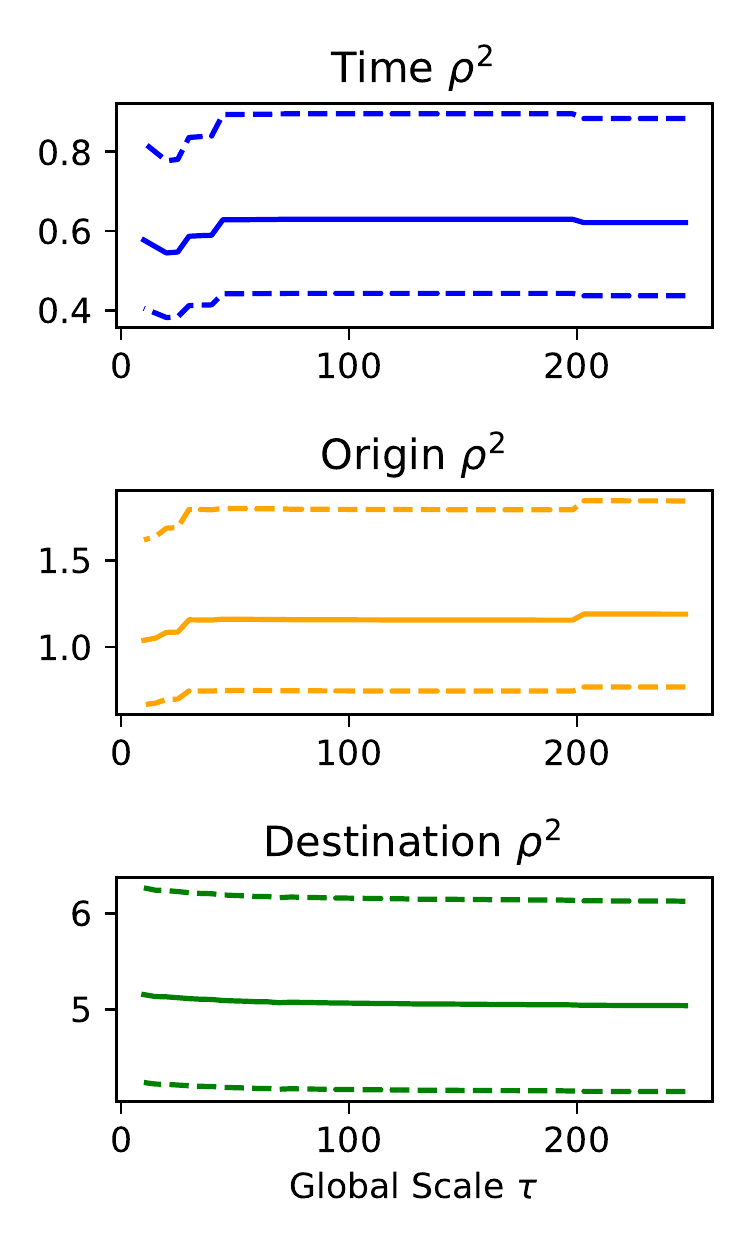}
\caption{Trajectory of Random Effects and Random Effects Variance parameters in Iraq case study.}
    \label{fig:app_spt_traj}
\end{figure}

\begin{figure}
    \centering
    \includegraphics[scale=0.8]{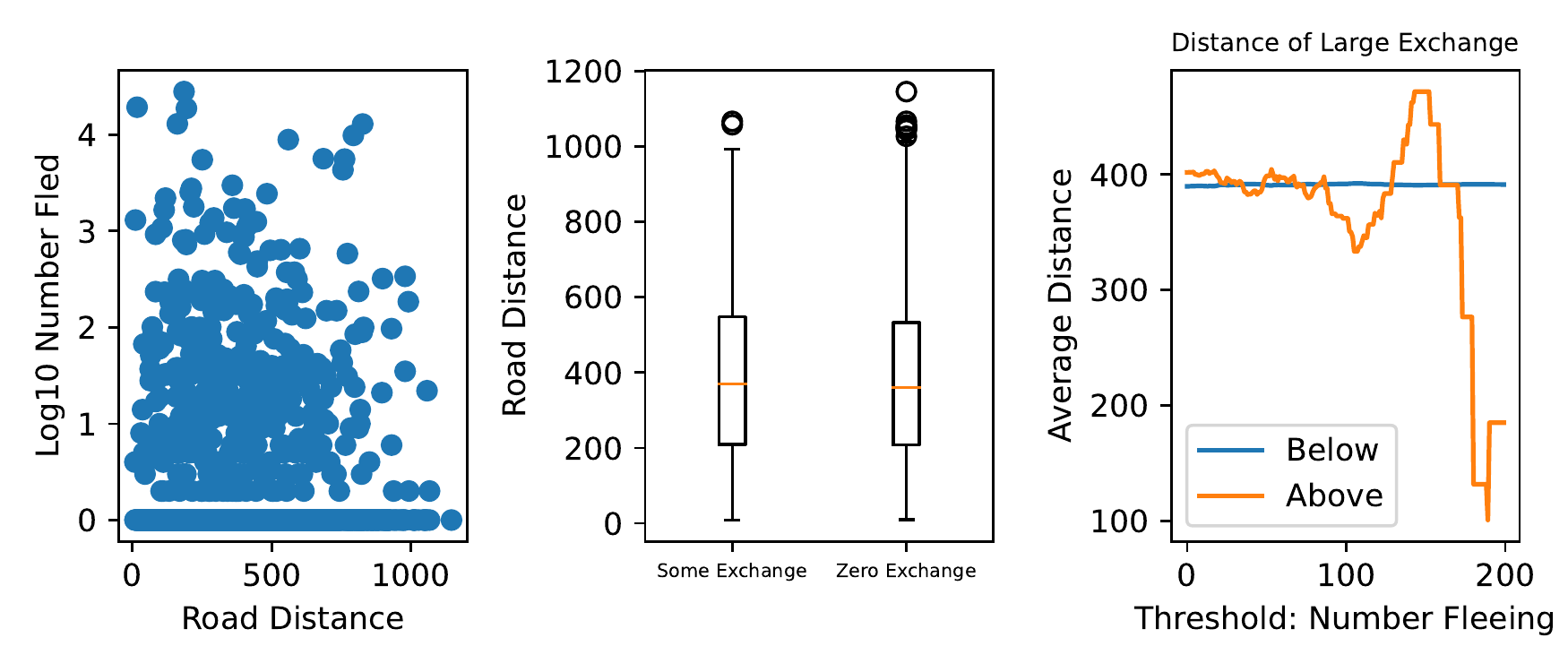}
    \caption{\textbf{Antigravitation:} \textit{Left:} Scatterplots do not reveal a clear relationship between pairwise distance and migrant exchange. \textit{Mid:} Boxplots don't either. \textit{Right:} y-axis gives the mean distance between those regions with a migrant exchange of at least the number given by the x-axis. For some exchange levels, this is above the overall overage, and for others, below.}
    \label{fig:antigravity}
\end{figure}

\subsection{The Antigravity Model}\label{sec:app_antigravity}

The gravity model is so named due to its consideration of the distances between exchanging regions, supposing that greater distance means less exchange.
But this might not be the case in subnational regions: we observe a negative coefficient, meaning exchange is, \textit{ceteris paribus}, increasingly likely with distance.
In a small country where one side of it is at war, it's not too surprising that migrants would want to go to the ``other side".
Without random effects and without penalties, distance still assumes a negative, though smaller, value.

\bibliographystyle{chicago}
\bibliography{main}

\begin{thebibliography}{}

\bibitem[\protect\citeauthoryear{Amari}{Amari}{1998}]{amari1998natural}
Amari, S.-I. (1998).
\newblock Natural gradient works efficiently in learning.
\newblock {\em Neural computation\/}~{\em 10\/}(2), 251--276.

\bibitem[\protect\citeauthoryear{Babacan, Nakajima, and Do}{Babacan
  et~al.}{2014}]{babacan2014bayesian}
Babacan, S.~D., S.~Nakajima, and M.~N. Do (2014).
\newblock Bayesian group-sparse modeling and variational inference.
\newblock {\em IEEE transactions on signal processing\/}~{\em 62\/}(11),
  2906--2921.

\bibitem[\protect\citeauthoryear{Beck and Teboulle}{Beck and
  Teboulle}{2009}]{beck2009fast}
Beck, A. and M.~Teboulle (2009).
\newblock A fast iterative shrinkage-thresholding algorithm for linear inverse
  problems.
\newblock {\em SIAM journal on imaging sciences\/}~{\em 2\/}(1), 183--202.

\bibitem[\protect\citeauthoryear{Bhadra, Datta, Polson, and Willard}{Bhadra
  et~al.}{2019}]{Bhadra2019}
Bhadra, A., J.~Datta, N.~G. Polson, and B.~Willard (2019).
\newblock Lasso meets horseshoe: A survey.
\newblock {\em Statistical Science\/}~{\em 34\/}(3), 405--427.

\bibitem[\protect\citeauthoryear{Bhattacharya, Pati, Pillai, and
  Dunson}{Bhattacharya et~al.}{2012}]{bhattacharya2012bayesian}
Bhattacharya, A., D.~Pati, N.~S. Pillai, and D.~B. Dunson (2012).
\newblock Bayesian shrinkage.
\newblock {\em arXiv preprint arXiv:1212.6088\/}.

\bibitem[\protect\citeauthoryear{Blei, Kucukelbir, and McAuliffe}{Blei
  et~al.}{2017}]{blei2017variational}
Blei, D.~M., A.~Kucukelbir, and J.~D. McAuliffe (2017).
\newblock Variational inference: A review for statisticians.
\newblock {\em Journal of the American statistical Association\/}~{\em
  112\/}(518), 859--877.

\bibitem[\protect\citeauthoryear{Bühlmann and Meier}{Bühlmann and
  Meier}{2008}]{buhlmann2008}
Bühlmann, P. and L.~Meier (2008).
\newblock {Discussion: One-step sparse estimates in nonconcave penalized
  likelihood models}.
\newblock {\em The Annals of Statistics\/}~{\em 36\/}(4), 1534 -- 1541.

\bibitem[\protect\citeauthoryear{Candes, Wakin, and Boyd}{Candes
  et~al.}{2008}]{candes2008enhancing}
Candes, E.~J., M.~B. Wakin, and S.~P. Boyd (2008).
\newblock Enhancing sparsity by reweighted l 1 minimization.
\newblock {\em Journal of Fourier analysis and applications\/}~{\em 14\/}(5),
  877--905.

\bibitem[\protect\citeauthoryear{Carvalho, Polson, and Scott}{Carvalho
  et~al.}{2010}]{Carvalho2010}
Carvalho, C.~M., N.~G. Polson, and J.~G. Scott (2010).
\newblock The horseshoe estimator for sparse signals.
\newblock {\em Biometrika\/}~{\em 97\/}(2), 465--480.

\bibitem[\protect\citeauthoryear{Chen, Mirestean, and Tsangarides}{Chen
  et~al.}{2018}]{Chen2018}
Chen, H., A.~Mirestean, and C.~G. Tsangarides (2018).
\newblock Bayesian model averaging for dynamic panels with an application to a
  trade gravity model.
\newblock {\em Econometric Reviews\/}~{\em 37\/}(7), 777--805.

\bibitem[\protect\citeauthoryear{Congdon}{Congdon}{2000}]{Congdon2000}
Congdon, P. (2000).
\newblock A {B}ayesian approach to prediction using the gravity model, with an
  application to patient flow modeling.
\newblock {\em Geographical analysis\/}~{\em 32\/}(3), 205--224.

\bibitem[\protect\citeauthoryear{Daubechies, Defrise, and De~Mol}{Daubechies
  et~al.}{2004}]{daubechies2004iterative}
Daubechies, I., M.~Defrise, and C.~De~Mol (2004).
\newblock An iterative thresholding algorithm for linear inverse problems with
  a sparsity constraint.
\newblock {\em Communications on Pure and Applied Mathematics\/}~{\em
  57\/}(11), 1413--1457.

\bibitem[\protect\citeauthoryear{De~Oliveira}{De~Oliveira}{2012}]{de2012bayesian}
De~Oliveira, V. (2012).
\newblock Bayesian analysis of conditional autoregressive models.
\newblock {\em Annals of the Institute of Statistical Mathematics\/}~{\em
  64\/}(1), 107--133.

\bibitem[\protect\citeauthoryear{Efron, Hastie, Johnstone, and
  Tibshirani}{Efron et~al.}{2004}]{Efron2004}
Efron, B., T.~Hastie, I.~Johnstone, and R.~Tibshirani (2004).
\newblock {Least angle regression}.
\newblock {\em The Annals of Statistics\/}~{\em 32\/}(2), 407 -- 499.

\bibitem[\protect\citeauthoryear{Fan and Li}{Fan and Li}{2001}]{Fan2001}
Fan, J. and R.~Li (2001).
\newblock Variable selection via nonconcave penalized likelihood and its oracle
  properties.
\newblock {\em Journal of the American Statistical Association\/}~{\em
  96\/}(456), 1348--1360.

\bibitem[\protect\citeauthoryear{Fan, Xue, and Zou}{Fan
  et~al.}{2014}]{fan2014strong}
Fan, J., L.~Xue, and H.~Zou (2014).
\newblock Strong oracle optimality of folded concave penalized estimation.
\newblock {\em Annals of statistics\/}~{\em 42\/}(3), 819.

\bibitem[\protect\citeauthoryear{Ferreira}{Ferreira}{2019}]{ferreira2019limiting}
Ferreira, M.~A. (2019).
\newblock The limiting distribution of the {G}ibbs sampler for the intrinsic
  conditional autoregressive model.
\newblock {\em Brazilian Journal of Probability and Statistics\/}~{\em
  33\/}(4), 734--744.

\bibitem[\protect\citeauthoryear{Fu and Knight}{Fu and
  Knight}{2000}]{fu2000asymptotics}
Fu, W. and K.~Knight (2000).
\newblock Asymptotics for lasso-type estimators.
\newblock {\em The Annals of statistics\/}~{\em 28\/}(5), 1356--1378.

\bibitem[\protect\citeauthoryear{Hare and Sagastiz{\'a}bal}{Hare and
  Sagastiz{\'a}bal}{2009}]{hare2009computing}
Hare, W. and C.~Sagastiz{\'a}bal (2009).
\newblock Computing proximal points of nonconvex functions.
\newblock {\em Mathematical Programming\/}~{\em 116\/}(1), 221--258.

\bibitem[\protect\citeauthoryear{Helgøy and Li}{Helgøy and
  Li}{2019}]{helgoy2019sparse}
Helgøy, I.~M. and Y.~Li (2019).
\newblock A {B}ayesian lasso based sparse learning model.
\newblock {\em arXiv preprint arXiv:1908.07220\/}.

\bibitem[\protect\citeauthoryear{Hoerl and Kennard}{Hoerl and
  Kennard}{1970}]{hoerl1970ridge}
Hoerl, A.~E. and R.~W. Kennard (1970).
\newblock Ridge regression: Biased estimation for nonorthogonal problems.
\newblock {\em Technometrics\/}~{\em 12\/}(1), 55--67.

\bibitem[\protect\citeauthoryear{Hunter and Li}{Hunter and
  Li}{2005}]{Hunter2005}
Hunter, D.~R. and R.~Li (2005).
\newblock {Variable selection using MM algorithms}.
\newblock {\em The Annals of Statistics\/}~{\em 33\/}(4), 1617 -- 1642.

\bibitem[\protect\citeauthoryear{Kang and Guo}{Kang and
  Guo}{2009}]{kang2009self}
Kang, J. and J.~Guo (2009).
\newblock Self-adaptive lasso and its {B}ayesian estimation.
\newblock Technical report, Working Paper.

\bibitem[\protect\citeauthoryear{Karemera, Oguledo, and Davis}{Karemera
  et~al.}{2000}]{Karemera2000}
Karemera, D., V.~I. Oguledo, and B.~Davis (2000).
\newblock A gravity model analysis of international migration to north america.
\newblock {\em Applied Economics\/}~{\em 32\/}(13), 1745--1755.

\bibitem[\protect\citeauthoryear{Kawano, Hoshina, Shimamura, and
  Konishi}{Kawano et~al.}{2015}]{kawano2015predictive}
Kawano, S., I.~Hoshina, K.~Shimamura, and S.~Konishi (2015).
\newblock Predictive model selection criteria for {B}ayesian lasso regression.
\newblock {\em Journal of the Japanese Society of Computational
  Statistics\/}~{\em 28\/}(1), 67--82.

\bibitem[\protect\citeauthoryear{Kingma and Ba}{Kingma and
  Ba}{2015}]{kingma2015adam}
Kingma, D.~P. and J.~Ba (2015).
\newblock Adam: A method for stochastic optimization.
\newblock In {\em ICLR (Poster)}.

\bibitem[\protect\citeauthoryear{Kingma and Welling}{Kingma and
  Welling}{2014}]{Kingma2014}
Kingma, D.~P. and M.~Welling (2014).
\newblock {Auto-Encoding Variational Bayes}.
\newblock In {\em 2nd International Conference on Learning Representations,
  {ICLR} 2014, Banff, AB, Canada, April 14-16, 2014, Conference Track
  Proceedings}.

\bibitem[\protect\citeauthoryear{Koltchinskii, Lounici, and
  Tsybakov}{Koltchinskii et~al.}{2011}]{koltchinskii2011nuclear}
Koltchinskii, V., K.~Lounici, and A.~B. Tsybakov (2011).
\newblock Nuclear-norm penalization and optimal rates for noisy low-rank matrix
  completion.
\newblock {\em The Annals of Statistics\/}~{\em 39\/}(5), 2302--2329.

\bibitem[\protect\citeauthoryear{Leng, Tran, and Nott}{Leng
  et~al.}{2014}]{leng2014bayesian}
Leng, C., M.-N. Tran, and D.~Nott (2014).
\newblock Bayesian adaptive lasso.
\newblock {\em Annals of the Institute of Statistical Mathematics\/}~{\em
  66\/}(2), 221--244.

\bibitem[\protect\citeauthoryear{Makalic and Schmidt}{Makalic and
  Schmidt}{2015}]{Makalic2015}
Makalic, E. and D.~F. Schmidt (2015).
\newblock A simple sampler for the horseshoe estimator.
\newblock {\em IEEE Signal Processing Letters\/}~{\em 23\/}(1), 179--182.

\bibitem[\protect\citeauthoryear{Mallick and Yi}{Mallick and
  Yi}{2014}]{mallick2014new}
Mallick, H. and N.~Yi (2014).
\newblock A new {B}ayesian lasso.
\newblock {\em Statistics and its interface\/}~{\em 7\/}(4), 571.

\bibitem[\protect\citeauthoryear{Marjanovic and Solo}{Marjanovic and
  Solo}{2013}]{marjanovic2013exact}
Marjanovic, G. and V.~Solo (2013).
\newblock On exact l q denoising.
\newblock In {\em 2013 IEEE International Conference on Acoustics, Speech and
  Signal Processing}, pp.\  6068--6072. IEEE.

\bibitem[\protect\citeauthoryear{Meinshausen}{Meinshausen}{2007}]{meinshausen2007relaxed}
Meinshausen, N. (2007).
\newblock Relaxed lasso.
\newblock {\em Computational Statistics \& Data Analysis\/}~{\em 52\/}(1),
  374--393.

\bibitem[\protect\citeauthoryear{Meyer}{Meyer}{2021}]{Meyer2021}
Meyer, G.~P. (2021, jun).
\newblock An alternative probabilistic interpretation of the huber loss.
\newblock In {\em 2021 IEEE/CVF Conference on Computer Vision and Pattern
  Recognition (CVPR)}, Los Alamitos, CA, USA, pp.\  5257--5265. IEEE Computer
  Society.

\bibitem[\protect\citeauthoryear{Mitchell and Beauchamp}{Mitchell and
  Beauchamp}{1988}]{Mitchell1998}
Mitchell, T.~J. and J.~J. Beauchamp (1988).
\newblock Bayesian variable selection in linear regression.
\newblock {\em Journal of the American Statistical Association\/}~{\em
  83\/}(404), 1023--1032.

\bibitem[\protect\citeauthoryear{Nesterov}{Nesterov}{1983}]{nesterov1983method}
Nesterov, Y.~E. (1983).
\newblock A method for solving the convex programming problem with convergence
  rate o (1/k\^{} 2).
\newblock In {\em Dokl. akad. nauk Sssr}, Volume 269, pp.\  543--547.

\bibitem[\protect\citeauthoryear{Owen}{Owen}{2007}]{owen2007robust}
Owen, A.~B. (2007).
\newblock A robust hybrid of lasso and ridge regression.
\newblock {\em Contemporary Mathematics\/}~{\em 443\/}(7), 59--72.

\bibitem[\protect\citeauthoryear{Owen}{Owen}{2013}]{Owen2013}
Owen, A.~B. (2013).
\newblock {\em Monte Carlo theory, methods and examples}.

\bibitem[\protect\citeauthoryear{Parikh, Boyd, et~al.}{Parikh
  et~al.}{2014}]{parikh2014proximal}
Parikh, N., S.~Boyd, et~al. (2014).
\newblock Proximal algorithms.
\newblock {\em Foundations and Trends in Optimization\/}~{\em 1\/}(3),
  127--239.

\bibitem[\protect\citeauthoryear{Park and Casella}{Park and
  Casella}{2008}]{Park2008}
Park, T. and G.~Casella (2008).
\newblock The {B}ayesian lasso.
\newblock {\em Journal of the American Statistical Association\/}~{\em
  103\/}(482), 681--686.

\bibitem[\protect\citeauthoryear{Polson, Scott, and Willard}{Polson
  et~al.}{2015}]{polson2015proximal}
Polson, N.~G., J.~G. Scott, and B.~T. Willard (2015).
\newblock Proximal algorithms in statistics and machine learning.
\newblock {\em Statistical Science\/}~{\em 30\/}(4), 559--581.

\bibitem[\protect\citeauthoryear{Poprawe}{Poprawe}{2015}]{Poprawe2015}
Poprawe, M. (2015, Jun).
\newblock On the relationship between corruption and migration: empirical
  evidence from a gravity model of migration.
\newblock {\em Public Choice\/}~{\em 163\/}(3), 337--354.

\bibitem[\protect\citeauthoryear{Quaini and Trojani}{Quaini and
  Trojani}{2022}]{quaini2022proximal}
Quaini, A. and F.~Trojani (2022).
\newblock Proximal estimation and inference.
\newblock {\em arXiv preprint arXiv:2205.13469\/}.

\bibitem[\protect\citeauthoryear{Ramos and Suriñach}{Ramos and
  Suriñach}{2017}]{Ramos2017}
Ramos, R. and J.~Suriñach (2017).
\newblock A gravity model of migration between the {ENC} and the {EU}.
\newblock {\em Tijdschrift voor Economische en Sociale Geografie\/}~{\em
  108\/}(1), 21--35.

\bibitem[\protect\citeauthoryear{Ranganath, Gerrish, and Blei}{Ranganath
  et~al.}{2014}]{Ranganath2014}
Ranganath, R., S.~Gerrish, and D.~Blei (2014, 22--25 Apr).
\newblock {Black Box Variational Inference}.
\newblock In S.~Kaski and J.~Corander (Eds.), {\em Proceedings of the
  Seventeenth International Conference on Artificial Intelligence and
  Statistics}, Volume~33 of {\em Proceedings of Machine Learning Research},
  Reykjavik, Iceland, pp.\  814--822. PMLR.

\bibitem[\protect\citeauthoryear{Ranjan and Tobias}{Ranjan and
  Tobias}{2007}]{Ranjan2007}
Ranjan, P. and J.~L. Tobias (2007).
\newblock Bayesian inference for the gravity model.
\newblock {\em Journal of Applied Econometrics\/}~{\em 22\/}(4), 817--838.

\bibitem[\protect\citeauthoryear{Robinson}{Robinson}{1996}]{Robinson1996}
Robinson, S.~M. (1996).
\newblock Analysis of sample-path optimization.
\newblock {\em Mathematics of Operations Research\/}~{\em 21\/}(3), 513--528.

\bibitem[\protect\citeauthoryear{Ročková and George}{Ročková and
  George}{2018}]{Rockova2018}
Ročková, V. and E.~I. George (2018).
\newblock The spike-and-slab lasso.
\newblock {\em Journal of the American Statistical Association\/}~{\em
  113\/}(521), 431--444.

\bibitem[\protect\citeauthoryear{Sandhu, Khalil, Pettit, Poirel, and
  Sarkar}{Sandhu et~al.}{2021}]{sandhu2021nonlinear}
Sandhu, R., M.~Khalil, C.~Pettit, D.~Poirel, and A.~Sarkar (2021).
\newblock Nonlinear sparse {B}ayesian learning for physics-based models.
\newblock {\em Journal of Computational Physics\/}~{\em 426}, 109728.

\bibitem[\protect\citeauthoryear{She}{She}{2009}]{she2009thresholding}
She, Y. (2009).
\newblock Thresholding-based iterative selection procedures for model selection
  and shrinkage.
\newblock {\em Electronic Journal of statistics\/}~{\em 3}, 384--415.

\bibitem[\protect\citeauthoryear{Taylor, Banks, and McCoy}{Taylor
  et~al.}{1979}]{taylor1979deconvolution}
Taylor, H.~L., S.~C. Banks, and J.~F. McCoy (1979).
\newblock Deconvolution with the $\ell_1$ norm.
\newblock {\em Geophysics\/}~{\em 44\/}(1), 39--52.

\bibitem[\protect\citeauthoryear{Terenin, Dong, and Draper}{Terenin
  et~al.}{2019}]{Terenin2019}
Terenin, A., S.~Dong, and D.~Draper (2019).
\newblock {GPU}-accelerated {G}ibbs sampling: a case study of the horseshoe
  probit model.
\newblock {\em Statistics and Computing\/}~{\em 29\/}(2), 301--310.

\bibitem[\protect\citeauthoryear{Tibshirani}{Tibshirani}{1996}]{Tibshirani1996}
Tibshirani, R. (1996).
\newblock Regression shrinkage and selection via the lasso.
\newblock {\em Journal of the Royal Statistical Society. Series B
  (Methodological)\/}~{\em 58\/}(1), 267--288.

\bibitem[\protect\citeauthoryear{Tibshirani, Saunders, Rosset, Zhu, and
  Knight}{Tibshirani et~al.}{2005}]{tibshirani2005sparsity}
Tibshirani, R., M.~Saunders, S.~Rosset, J.~Zhu, and K.~Knight (2005).
\newblock Sparsity and smoothness via the fused lasso.
\newblock {\em Journal of the Royal Statistical Society: Series B (Statistical
  Methodology)\/}~{\em 67\/}(1), 91--108.

\bibitem[\protect\citeauthoryear{Tipping}{Tipping}{2001}]{tipping2001sparse}
Tipping, M.~E. (2001).
\newblock Sparse {B}ayesian learning and the relevance vector machine.
\newblock {\em Journal of machine learning research\/}~{\em 1\/}(Jun),
  211--244.

\bibitem[\protect\citeauthoryear{Tung, Tran, and Cuong}{Tung
  et~al.}{2019}]{tung2019bayesian}
Tung, D.~T., M.-N. Tran, and T.~M. Cuong (2019).
\newblock Bayesian adaptive lasso with variational {B}ayes for variable
  selection in high-dimensional generalized linear mixed models.
\newblock {\em Communications in Statistics-Simulation and Computation\/}~{\em
  48\/}(2), 530--543.

\bibitem[\protect\citeauthoryear{{Wilson Center}}{{Wilson
  Center}}{2019}]{wilson2019Timeline}
{Wilson Center} (2019).
\newblock Timeline: the rise, spread, and fall of the islamic state.

\bibitem[\protect\citeauthoryear{Yuan and Lin}{Yuan and
  Lin}{2006}]{yuan2006model}
Yuan, M. and Y.~Lin (2006).
\newblock Model selection and estimation in regression with grouped variables.
\newblock {\em Journal of the Royal Statistical Society: Series B (Statistical
  Methodology)\/}~{\em 68\/}(1), 49--67.

\bibitem[\protect\citeauthoryear{Zhang}{Zhang}{2010}]{Zhang2010}
Zhang, C.-H. (2010).
\newblock {Nearly unbiased variable selection under minimax concave penalty}.
\newblock {\em The Annals of Statistics\/}~{\em 38\/}(2), 894 -- 942.

\bibitem[\protect\citeauthoryear{Zou}{Zou}{2006}]{Zou2006}
Zou, H. (2006).
\newblock The adaptive lasso and its oracle properties.
\newblock {\em Journal of the American statistical association\/}~{\em
  101\/}(476), 1418--1429.

\bibitem[\protect\citeauthoryear{Zou and Hastie}{Zou and
  Hastie}{2005}]{zou2005regularization}
Zou, H. and T.~Hastie (2005).
\newblock Regularization and variable selection via the elastic net.
\newblock {\em Journal of the royal statistical society: series B (statistical
  methodology)\/}~{\em 67\/}(2), 301--320.

\bibitem[\protect\citeauthoryear{Zou and Li}{Zou and Li}{2008}]{Zou2008}
Zou, H. and R.~Li (2008).
\newblock {One-step sparse estimates in nonconcave penalized likelihood
  models}.
\newblock {\em The Annals of Statistics\/}~{\em 36\/}(4), 1509 -- 1533.

\end{thebibliography}


\end{document}